\documentclass[12pt]{article}
\usepackage[intlimits]{amsmath}
\usepackage{amssymb,amsthm,slashed,mathabx}

\usepackage[T1]{fontenc}
\usepackage[cp1250]{inputenc}

\newcommand{\<}{\langle}
\renewcommand{\>}{\rangle}

\newcommand{\Sc}{\mathcal{S}}
\newcommand{\F}{\mathcal{F}}
\newcommand{\B}{\mathcal{B}}
\newcommand{\C}{\mathcal{C}}
\newcommand{\Hc}{\mathcal{H}}
\newcommand{\hM}{\widehat{M}}
\newcommand{\mR}{\mathbb{R}}

\newcommand{\clc}{\overline{V_+}}
\newcommand{\lc}{C_+}
\newcommand{\lct}{C_+^t}
\newcommand{\vq}{\vec{q}}
\newcommand{\vp}{\vec{p}}
\newcommand{\vx}{\vec{x}}

\newcommand{\vl}{\vec{l}}
\newcommand{\vz}{\vec{z}}
\newcommand{\hq}{\hat{q}}
\newcommand{\hn}{\hat{n}}
\newcommand{\hp}{\hat{p}}
\newcommand{\hx}{\hat{x}}

\newcommand{\df}{\dot{f}}

\newcommand{\nrm}[1]{{\left\vert\kern-0.25ex\left\vert\kern-0.25ex\left\vert #1
    \right\vert\kern-0.25ex\right\vert\kern-0.25ex\right\vert}}

\newcommand{\out}{\mathrm{out}}
\newcommand{\al}{\alpha}
\newcommand{\ga}{\gamma}
\newcommand{\ep}{\epsilon}
\newcommand{\la}{\lambda}

\newcommand{\dsp}{\displaystyle}

\newcommand{\vep}{\varepsilon}
\newcommand{\vph}{\varphi}
\newcommand{\vth}{\vartheta}
\newcommand{\ti}{\widetilde}
\newcommand{\w}{\omega}
\newcommand{\W}{\Omega}

\newcommand{\wh}{\widehat}
\newcommand{\wch}{\widecheck}
\newcommand{\p}{\partial}
\newcommand{\con}{\mathrm{const}}
\newcommand{\be}{\begin{equation}}
\newcommand{\ee}{\end{equation}}

\usepackage[showonlyrefs]{mathtools}
\mathtoolsset{showonlyrefs=true}
\usepackage{fixltx2e}
\MakeRobust{\eqref}

\usepackage{bbm}
\DeclareMathOperator{\id}{\mathbbm{1}}
\newcommand{\mN}{\mathbb{N}}
\newcommand{\Sdeg}[1]{\Sc^{#1}}

\DeclareMathOperator{\supp}{supp}

\DeclareMathOperator{\sgn}{sgn}
\newtheorem{thm}{Theorem}
\newtheorem{pr}[thm]{Proposition}
\newtheorem{lem}[thm]{Lemma}

\newtheorem{ass}{Assumption}
\newtheorem{dfn}{Definition}
\title{Massless asymptotic fields\\ and Haag-Ruelle theory}
\author{Pawe{\l} Duch\thanks{e-mail: pawel.duch@uj.edu.pl} {} and Andrzej Herdegen\thanks{e-mail: andrzej.herdegen@uj.edu.pl}\\
{\it Institute of Physics, Jagiellonian University,}\\
{\it ul.\,S.\,{\L}ojasiewicza 11, 30-348  Krak\'{o}w, Poland}}
\date{}

\begin{document}

\maketitle

\begin{abstract}\noindent
 We revisit the problem of the existence of asymptotic massless boson fields in quantum field theory. The well-known construction of such fields by Buchholz \cite{bu77}, \cite{bu82} is based on locality and the existence of vacuum vector, at least in regions spacelike to spacelike cones. Our analysis does not depend on these assumptions and supplies a more general framework for fields only very weakly decaying in spacelike directions. In this setting the existence of appropriate null asymptotes of fields is linked with their spectral properties in the neighborhood of the lightcone. The main technical tool is one of the results of a~recent analysis by one of us \cite{he14'}, which allows application of the null asymptotic limit separately to creation/annihilation parts of a wide class of non-local fields. In vacuum representation the scheme allows application of the methods of the Haag-Ruelle theory closely analogous to those of the massive case. In local case this Haag-Ruelle procedure may be combined with
the Buchholz method, which leads to significant simplification.

\vspace{1ex}
\noindent
keywords: quantum field theory, nonlocal fields, scattering theory, Haag-Ruelle theory

\vspace{1ex}

\end{abstract}

\section{Introduction}\label{int}

This work is a sequel to an earlier article by one of us \cite{he14}, in which the infraparticle problem was investigated as the question of the existence of asymptotic fields with energy-momentum transfer on the mass hyperboloid. The idea was tested on a model of asymptotic fields in electrodynamics.

In the present paper we turn attention to massless fields. In quantum electrodynamics photons do not carry electric charge, so they are not plagued by the infraparticle problem. As long as their energy-content stays away from zero (or at least vanishes fast enough in the neighborhood of zero), they pose no problems and may be interpreted as decent zero-mass particles described by Fock representation built in standard way on the vacuum state. However, this picture, although accepted without much ado in most standard textbooks on quantum field theory, breaks down as soon as interaction with charged particles is turned on.

A symptom of the problem appears already at the classical level: if a~particle carrying charge $q$ is scattered from an incoming four-velocity $v$ to an outgoing four-velocity $u$, then the radiation (retarded $-$ advanced) field produced in this process has a long-range tail $F_\mathrm{l.r.}$ of the Coulomb decay rate:
\begin{equation*}
  F_\mathrm{l.r.}(x)=2q\,x\wedge\bigg(\frac{v}{[(v\cdot x)^2-x^2]^{3/2}}
 -\frac{u}{[(u\cdot x)^2-x^2]^{3/2}}\bigg)\,,\quad x^2<0\,.
\end{equation*}
However, such nonzero tail produces in the momentum space the small-energy content which excludes the possibility of representing the field, upon quantization, as a photonic state in the standard Fock space. This is reflected in the appearance of infrared divergencies in the perturbative quantum electrodynamics. On a more elevated level of general structural analysis, as considered within the algebraic approach to QFT, one realizes that the Fock space picture is insufficient for an adequate description of state space of quantum electrodynamics.

In this article we obtain, in a wide context, a relation between null asymptotic behavior of fields and their spectral properties in the neighborhood of the lightcone in the energy-momentum space. Fields are assumed to satisfy some weak condition on the decay of their commutators in spacelike directions, but need not be local. We believe that the limits of strict locality may be to narrow in a constrained theory like QED; we refer the reader to our earlier papers \cite{he14}, \cite{he14'} for more comments on nonlocality and our specific choice of commutator laws. The smearing applied here to the fields still cuts the infrared-singular degrees of freedom, but no further assumptions beside relativistic positivity of energy spectrum in the Hilbert space are needed. In particular, we do not assume the existence of vacuum vector. The main tool for this analysis is a recently obtained estimate on the norms of creation/annihilation components of fields in the assumed class~\cite{he14'}. If a nonzero asymptotic
limit of a field exists, it defines a field with energy-momentum transfer on the lightcone, satisfying the wave equation. Its creation/annihilation components describe particle excitations, but the particle interpretation need not be complete, even in regions spacelike to spacelike cones, as considered in \cite{bu82}.

Next, we show that in vacuum representation the scheme allows the application of the Haag-Ruelle procedure in close analogy with the massive case. One needs here a spectral condition of the Herbst type, as well as some weak clustering condition (satisfied in local case). In strictly local case the spectral condition is redundant, if one follows the method used by Buchholz~\cite{bu77}. The proof of the existence of asymptotic fields and scattering states is then greatly simplified as compared to \cite{bu77}, due to the independent bounds on creation/annihilation components.

Our notation, mostly standard, is as in the article \cite{he14'}, which is a prerequisite for the present analysis. Here we recall only a few basic conventions: $M$ is the Minkowski vector space (`configuration space' with fixed central point), $\hM$ is the `momentum space' isometric with $M$. The unit, future-pointing vector of a chosen time-axis is denoted by $t$; $\vx$ is the $3$-space part of the vector $x$, and $|x|^2=|x^0|^2+|\vx|^2$. Moreover, for $x\in M$, $p\in\hM$, we shall denote
\begin{equation}\label{hats}
 \hx=\frac{\vx}{|\vx|}\,, \qquad \hp=\frac{\vp}{|\vp|}\,.
\end{equation}
Our conventions and notation for Fourier transforms are:
\begin{equation*}
\begin{gathered}
 \ti{f}(\w)=(\F_1f)(\w)=\frac{1}{2\pi}\int e^{i\w s} f(s)ds\,,\\
 (\F_3{g})(\vp)=\frac{1}{2\pi}\int e^{-i\vp\cdot\vx} g(\vx)d^3x\,,\\
 \wh{\chi}(p)=(\F{\chi})(p)=\frac{1}{(2\pi)^2}\int e^{ip\cdot x} \chi(x)dx\,,\quad
 \wch{\vph}=\F^{-1}\vph\,.
\end{gathered}
\end{equation*}
Throughout the article $\la$ is a~fixed parameter of the physical dimension of length.

We devote a few words to the lightcone geometry. The term `lightcone' itself is ambiguous, thus we specify: the solid, closed future lightcone will be denoted by $\clc$, and for the set of (nonzero) future null vectors we shall write:
\begin{equation}
 \lc=\{\,l\mid l\cdot l=0,\ l^0>0\,\}\,.
\end{equation}
It will also prove useful to denote
\begin{equation}
 \lct=\lc\cap\{\,l\mid t\cdot l=1\,\}\,,
\end{equation}
which may be interpreted as $t+S^2$, where $S^2$ is the unit sphere in the space orthogonal to $t$. We also recall that if $f$ is a complex measurable function on $\lc$, which in addition is homogeneous of degree $-2$: $f(\ga l)=\ga^{-2}f(l)$, then the integral defined by
\begin{equation}
 \int f(l)\,d^2l=\int_{\lct} f(l)\,d\W_t(l)\,,
\end{equation}
where $d\W_t(l)$ is the angle measure on the unit sphere, does not depend on the choice of the vector $t$ (Lorentz invariance; see, e.g., \cite{pr84}, \cite{he95}). Finally, we shall denote
\begin{equation}
 L_{ab}=l_a\frac{\p}{\p l^b}-l_b\frac{\p}{\p l^a}
\end{equation}
-- intrinsic differential operators on the lightcone, and recall that
\[
 \int L_{ab}f(l)\,d^2l=0\,.
\]

\section{Null asymptotes}

In what follows the following functional spaces will play important role.
\begin{dfn} Let $\ep>0$. We define the following vector spaces: \\
 (i) $\Sc_\ep$: the space of smooth complex functions on $\mR\times\lct$ satisfying the bounds
\begin{equation}
 |L_{a_1b_1}\ldots L_{a_kb_k}\p_s^mf(s,l)|\leq\frac{\con_{k,m}}{(\la+|s|)^{m+1+\ep}}\,,\qquad a_i,b_i\in\{1,2,3\}\,.
\end{equation}
This space depends on the choice of $t$ (not to burden notation we do not make this dependence explicit).\\
(ii) $\Sc_\ep^n$: the space of complex functions on $\mR\times\lc$ which are homogeneous of degree $n$: $f(\ga s,\ga l)=\ga^nf(s,l)$, and which satisfy the bounds
\begin{equation}
 |L_{a_1b_1}\ldots L_{a_kb_k}\p_s^mf(s,l)|\leq\frac{\con_{t,k,m}\,(t\cdot l)^{n-m}}{(\la+|s|/t\cdot l)^{m+1+\ep}}\,,\qquad a_i,b_i\in\{0,1,2,3\}\,.
\end{equation}
This space does not depend on the choice of vector $t$ (only bounding constants change).

Each function in $\Sc_\ep$ may be extended to a function in $\Sc_\ep^n$ by homogeneity. Conversely, each function in $\Sc_\ep^n$ belongs to $\Sc_\ep$ when restricted to $\mR\times\lct$; such restriction will be called the $t$-gauge.
\end{dfn}

We recapitulate some facts on classical fields satisfying the wave equation (in the form presented in \cite{he95}). Solutions which may be obtained as radiation fields (i.e.\ retarded minus  advanced fields) of some matter source (point particles or massive fields) stabilizing in remote past and future may be written in a convenient integral representation as
\be\label{wave}
 B(x)=-\frac{1}{2\pi}\int \dot{b}(x\cdot l,l)\,d^2l\,,
\ee
where $\dot{b}(s,l)\in\Sc_\ep^{-2}$ for some $\ep>0$ (in electrodynamics this applies to the potential in Lorenz gauge). It follows that $\dot{b}(s,l)=\p b(s,l)/\p s$, where $b$ is smooth, homogeneous of degree $-1$, and has finite limits $b(\pm\infty,l)$ (an overdot will always denote the $s$-derivative).  Field $B(x)$ has well-defined null asymptotes:
\be\label{null}
 \lim_{r\to \infty}r B(x\pm rl)=\pm b(x\cdot l,l)\mp b(\pm\infty,l)\equiv b_{\begin{smallmatrix}\mathrm{out}\\ \mathrm{in}\end{smallmatrix}}(x\cdot l,l)
\ee
and its long-range spacelike behavior is given by ($y^2<0$)
\be\label{rtail}
 \lim_{r\to\infty}r B(x+ry)=-\frac{1}{2\pi}\int\Delta b(l)\,\delta(y\cdot l)\,d^2l\,,
\ee
where $\Delta b(l)=b(+\infty,l)-b(-\infty,l)$ and $\delta$ is the Dirac measure. Function $b(s,l)$ is defined up to an addition of an $s$-independent term; specifying $b(-\infty,l)=-b(+\infty,l)$ one makes it unique.

As the above useful representation does not seem to be widely known, we make some additional comments on its relation to more standard knowledge. Let $B(x)$ be a solution of the wave equation. Choose any time axis $st$, $s\in\mR$, $t$~a~unit, future-pointing vector, and draw the past lightcone with vertex in $s_0t$ on the axis. For $x$ inside this cone the Kirchhoff formula (see, e.g., \cite{pr84}) gives the value of $B(x)$ in terms of the values of this field on the cone (solution of the null Cauchy problem). If the field has well defined future null asymptotes, then taking the limit $s_0\to\infty$ one arrives at Eq.\,\eqref{wave} with $b$ as defined in \eqref{null}. We draw attention of the reader to the fact that the spacelike $1/r$-tail of such fields is even in the argument. One shows that fields with odd spacelike $1/r$-tail do exist and may also be represented by the integral of the form \eqref{wave}, but with $b$ not satisfying the decay properties; these fields do not have well-defined null
asymptotes. Fortunately, such fields are not produced in scattering processes. (It is easy to see that the Lorentz potential of the field $F_\mathrm{l.r.}$ given in Introduction is even.) The whole picture may be reflected in time.

Another way to view representation \eqref{wave}, \eqref{null} is by its relation to the Fourier representation. Let $\ti{\dot{b}}(\w,l)$ be the Fourier transform of $\dot{b}(s,l)$ in $s$, as defined in Introduction. The $(-2)$-homogeneity of $\dot{b}$ implies the scaling property
 $\ti{\dot{b}}(\ga^{-1}\w,\ga\, l)=\ga^{-1}\ti{\dot{b}}(\w,l)$, therefore the relation $c(\w l)=-\ti{\dot{b}}(\w,l)/\w$ is a consistent definition of a function $c(k)$ on the lightcone. Expressing $\dot{b}$ in \eqref{wave} in terms of its transform $\ti{\dot{b}}$ one can obtain the usual Fourier representation of the field:
\be
 B(x)=\frac{1}{\pi}\int c(k)\delta(k^2)\sgn(k^0)e^{-ik\cdot x}d^4k\,.
\ee
Note that $\Delta b(l)=\ti{\dot{b}}(0,l)/(2\pi)$; infrared-singular fields are those for which this function does not vanish, and produces a $1/r$-spacelike tail (as given by~\eqref{rtail}). In those cases $c(k)$ is singular in $k=0$, but $\w c(\w l)$ is a regular function of $\w$. Again, $c(k)$ becomes more singular for fields with odd spacelike tails.

Relations \eqref{wave} and \eqref{null} show that the field $B(x)$ may be reconstructed from its null (future \emph{or} past) asymptote . This is to be compared with the massive case, where a (sufficiently regular) free field $\psi(x)$ may be reconstructed from timelike (future \emph{or} past) asymptotic (i.e. for $\la\to\infty$) behavior  of $\la^{3/2}\psi(\pm\la v)$, functions of $v$ on the unit hyperboloid \mbox{$H_+=\{v\mid v^2=1, v^0>0\}$}. For interacting fields one expects that in remote future (and past) they separate into free fields, so that their appropriate asymptotes may serve to define free asymptotic fields. Although this separation is in fact not complete in a constrained theory like electrodynamics, it holds true for infrared-regular components.

For quantum fields, even if they are regular functions of $x$ as is the case for translations of a large class of bounded operators, limiting is singular and one needs preparatory smearing. In general, if $B(x)$ is a bounded and continuous classical or quantum field and $\nu$ is a complex Borel measure, then we denote
\be
 B(\nu)=\int B(x)\,d\nu(x)\,.
\ee
In particular, if $\chi$ is an integrable function (with respect to $dx$), then $B(\chi)=\int B(x)\chi(x)dx$; for future use we also note that  $B(\chi)(\nu)=B(\nu)(\chi)=B(\chi*\nu)$.

The above discussion of classical asymptotes supplies now possible choices for asymptotic smearings. The extraction of massive contributions may be expected to result from limiting in $\la$ of the operator $\la^{3/2}\int \psi(\la v)g(v)\,d\mu(v)$, where $g$ is a test function on $H_+$ and $d\mu(v)$ is the standard Lorentz-invariant measure on $H_+$ -- in this case the measure $\nu$ is supported on the hyperboloid $x^2=\la^2$, $x^0>0$. Discussion along these lines of fields (anti-)commuting asymptotically in spacelike directions was the subject of \cite{he14}. We found a~general relation between asymptotic behavior of this limit and spectral properties of the field in the inside of the lightcone in momentum space; no further assumptions were needed. In case of vacuum representation the scheme yields a~Lorentz-invariant formulation of the Haag-Ruelle theory. Moreover, the scheme has been shown to work in a model of
asymptotic electrodynamics.

In the present article we analyze smeared null asymptotes. We choose a~time axis crossing the origin, with unit future vector $t$. For $(s,l)\in\mR\times\lct$  the set of points $st+rl$ represents a timelike cylinder, product of a sphere of radius $r$ and the time axis. Parameter $s$ is the retarded time of a point on this cylinder. Let $r$ be large enough for this cylinder to be well outside the interaction region. Then smearing a field over a patch on this cylinder may be interpreted as measuring `radiation' going out into a chosen solid angle over a~defined retarded time-span (similarly -- incoming radiation and advanced time for a picture reversed in time; in both cases $\nu$ is supported on the cylinder).

Motivated by this discussion we introduce:
\begin{dfn}\label{smeared}
 Let $B(x)$ be a bounded classical or quantum continuous field. Choose $t$ and let $f\in\Sc_\ep$.  Then we denote
\be\label{smearf}
 B[r,f]=\frac{r}{2\pi}\int B(st+rl)f(s,l)\,ds\,d\W_t(l)\,.
\ee
Moreover, let $g$ be a real continuous function with compact support in $(0,\infty)$. Then we shall write
\be\label{smearg}
 B[g,f]=\int g(r)B[r,f]\,dr\,.
\ee
\end{dfn}\noindent
Both $B[r,f]$ and $B[g,f]$ depend on $t$ (not to burden notation we do not indicate this explicitly), but the limit asymptote for classical wave equation solution is invariant:
\begin{pr}
 If $B(x)$ is the classical free field \eqref{wave} and $f\in\Sc_\ep^{-2}$, then
\be
 \lim_{r\to\infty}B[r,f]=\int b_\mathrm{out}(s,l)f(s,l)ds\,d^2l\,,
\ee
which is a Lorentz-invariant quantity.
\end{pr}\noindent
An easy proof of this fact is based on relation \eqref{null} and on the following fact: if $F\in\Sc_\ep^{-3}$, then the integral $\int F(s,l)ds$ is a function of $l$ homogeneous of degree~$-2$. Therefore, the integral $\int F(s,l)ds\,d^2l$ is a Lorentz-invariant quantity. We shall see later, that similar invariance holds in quantum case.

\section{Asymptotic relations}\label{asymp}

We assume that a QFT is defined in terms of a field *-algebra of bounded operators acting in a Hilbert space $\Hc$. The algebra includes, beside observables, also operators interpolating between inequivalent representations of observables, such as creators/annihilators of electric charge. Spacetime translations are performed by a unitary, continuous representation $U(a)$ of the translation group acting in $\Hc$, and the spectrum of its generators is contained in $\clc$ (relativistic energy positivity). However, we do \emph{not} assume the existence of the vacuum vector state, nor the action of a Lorentz group representation in $\Hc$ . For each bounded operator $B$ acting in $\Hc$ one defines the field $B(x)=U(x)BU(-x)$. We shall write $B\in\C^n$ (resp.\ $B\in\C_t^n$) if all derivatives $D^\alpha B(x)$ with $|\alpha|\leq n$ (resp. $\p_0^lB(x^0)$ with $l\leq n$, in the Minkowski basis in which $t$ is the timelike basis vector) exist and are continuous in the norm sense.

Our analysis will be based on the following decay property.
\begin{dfn}[\cite{he14'}]\label{com}
We shall say that the commutator $[B_1,B_2]$ of bounded operators \mbox{$B_1,B_2$} is of $\kappa$-type, $\kappa>0$, if the following bound is satisfied:
\begin{equation}\label{fcom}
\begin{aligned}
 &\|[B_1,B_2(a)]\|\leq c D_{\kappa}(a)\,,\\
 &D_{\kappa}(a)\equiv \begin{cases}1 & a^2\geq0\\
                          \dfrac{\la^\kappa}{(\la+|\vec{a}|-|a^0|)^{\kappa}} & a^2<0\,.
             \end{cases}
\end{aligned}
\end{equation}
with some constant $c$ depending on $B_i$. The assumption is covariant: if the bound holds in any particular reference system, it is valid in all other, with some other constants $c$.

We shall say that $[B_1,B_2]$ is of $\kappa^\infty$-type (resp.\ $\kappa_t^\infty$-type), if \mbox{$B_i\in\C^\infty$} (resp.\ $B_i\in\C_t^\infty$) and all $[D^{\al_1}B_1,D^{\al_2}B_2]$  (resp.\ all $[\p_0^{n_1}B_1,\p_0^{n_2}B_2]$) are of $\kappa$-type.
\end{dfn}

For any real $k>0$ and $B\in\C_t^\infty$ operators $B^k_\pm$ introduced by
\begin{equation}\label{Bkdef}
 \wch{B^k_\pm}(p)=e^{\mp ik\pi/2}\theta(\pm p^0)\,|p^0|^k\,\wch{B}(p)
\end{equation}
are well defined, bounded and in $\C^\infty_t$ \cite{he14'} (Fourier transform conventions as defined in Introduction).
We shall denote by $G_\pm(E)$ any real functions $\<0,+\infty)\mapsto\<0,+\infty\>$ which satisfy the following conditions:
\begin{equation}\label{GE}
\begin{split}
 \text{(i)}\quad &G_\pm(E)\ \text{are nonincreasing}\,,\\
 \text{(ii)}\quad &G_\pm\in L^2(\<0,+\infty))\,,\\
 \text{(iii)}\quad  &G_+(E)\leq\con\,
\end{split}
\end{equation}
(note that $G_-(0)$ may take the value $+\infty$). Then the following holds.
\begin{thm}[\cite{he14'}]\label{GB}\ \\
(i) If $[B,B^*]$ is of $\kappa_t^\infty$-type, then for $k>(\kappa+1)/2$  one has the bound
\begin{equation}
 \|B^k_\pm(\nu)G_\pm(P^0)\|^2\leq\con\int D_\kappa(x-y)\,d|\nu|(x)\,d|\nu|(y)\,,\label{boundpm}\\
\end{equation}
with the bounding constant depending on $B$, $\kappa$, $k$ and $G_\pm$, and where $|\nu|$ is the variation measure of the complex Borel measure $\nu$.\\
(ii) If $[B_1,B_2]$ is of $\kappa_t^\infty$-type, then for $k\geq\kappa$ also $[B_1,B^k_{2\pm}]$ and $[B^k_{1\pm},B^k_{2\pm}]$ (uncorrelated signs) are of $\kappa_t^\infty$-type.
\end{thm}
\noindent
We note that the classes of operators $G_\pm(P^0)$ include $(1+\la P^0)^{-1}$, which will be of particular interest below, starting from Eq.\,\eqref{PBP}.

Using this tool we shall now obtain the asymptotic behavior of $B$ smeared according to \eqref{smearf} in Definition~\ref{smeared}. Moreover, it is a well-known fact that one can usually regularize asymptotic behavior of fields by additionally smearing them in limiting parameter. Thus we assume that $g$ is smooth, such that
\begin{equation}\label{gsupint}
 \supp g\subseteq\<\tau_1,\tau_2\>\subset (0,\infty)\,,\qquad \int g(r) dr=1\,,
\end{equation}
and for $\eta\in(0,1\>$, $w=w(R)=\la(R/\la)^\eta$, we define
\be\label{geta}
 g^\eta_R(r)=w^{-1}g(w^{-1}(r-R)+1)\,,\quad g_R(r)=g_R^1(r)=R^{-1}g(R^{-1}r)\,.
\ee
These functions serve to smear $B$ as in Eq.\,\eqref{smearg}.
\begin{thm}\label{limit}
Let $[B,B^*]$ and $[B_1,B_2]$ be of $\kappa_t^\infty$-type, and let $f,f_1,f_2\in\Sc_\ep$ with $\ep>2$. Then\\
(i) for $\kappa<2$, $k>(\kappa+1)/2$:
\begin{equation}
  \|B^k_\pm[r,f]G_\pm(P^0)\|^2\leq\con(\la+r)^{2-\kappa}\,;
\end{equation}
(ii) for $\kappa>2$, $k>3/2$:
\begin{multline}\label{flimit}
 \limsup_{R\to\infty}\|B^k_\pm[g_R^\eta,f]G_\pm(P^0)\|^2\leq\con\,\limsup_{r\to\infty}\|B^k_\pm[r,f]G_\pm(P^0)\|^2\\
 \leq\con\int|f(s_1,l)f(s_2,l)|[(s_1-s_2)^2+\la^2]\,ds_1ds_2\,d\W_t(l)\,.
\end{multline}
(iii) For $k\geq\kappa>2$ let the supports of the functions
\be
 \lct\ni l\mapsto\|f_i(.,l)\|_\infty=\sup_{s\in\mR}|f_i(s,l)|\,,\quad i=1,2
\ee
be disjoint. Then for $\beta=\min\{\kappa,\ep\}$ and any $c\geq1$ the bound
\begin{equation}\label{disj}
 \big\|\big[B^k_{1\pm}{}[r_1,f_1],B^k_{2\pm}[r_2,f_2]\big]\big\|\leq\frac{\con}{(r_1r_2)^{(\beta-2)/2}}
\end{equation}
(uncorrelated signs) holds uniformly for $1/c\leq r_1/r_2\leq c$. Moreover, for any functions $g^{\eta_i}_i$, $i=1,2$:
\be\label{disjR}
  \big\|\big[B^k_{1\pm}{}[g_{1R}^{\eta_1},f_1],B^k_{2\pm}[g_{2R}^{\eta_2},f_2]\big]\big\|=O(R^{-(\beta-2)})\quad (R\to\infty)\,.
\ee
\end{thm}
\begin{proof}
We denote $r^2=r_1r_2$, $\Delta r=r_2-r_1$, $\Delta s=s_2-s_1$, $\xi^2=l_1\cdot l_2/2$; as $l_1,l_2\in C_+^t$, there is $\xi\in\<0,1\>$. Then $|r_2\vec{l}_2-r_1\vec{l}_1|^2=(\Delta r)^2+4r^2\xi^2$. In case (ii) there is $2k-1>2$, so without restricting generality we can restrict attention to $\kappa$ such that $2k-1>\kappa>2$. Now we can use Theorem~\ref{GB}: for cases (i) and (ii) point (i) of the Theorem, and for case (iii) point (ii). Thus in each of these cases we have to find bounds on $I=\int I(s_1,l_1)|f_1(s_1,l_1)|\,ds_1d\W_t(l_1)$, with
\begin{equation}\label{I1}
 I(s_1,l_1)=r^2\int D_\kappa(\Delta s+\Delta r,r_2\vec{l}_2-r_1\vec{l}_1)|f_2(s_2,l_2(\xi,\vph))|d\xi^2 d\vph ds_2\,,
\end{equation}
where $\vph$ is the azimuthal angle of $\vec{l}_2$ in the plane orthogonal to $\vec{l}_1$, and where $f_1=f_2=f$, $r_1=r_2=r$ in case (i) and (ii), and $|\Delta r|\leq dr$ with $d=c-1/c$ in case~(iii).\\
Case (i) and (ii). We change the integration variable in \eqref{I1} by $\rho=2r\xi$ and split integration region into two regions: (a)~$\rho\leq|\Delta s|$, and (b) the rest. In region (a) we use $D_\kappa\leq1$ and find
\begin{equation}
 I(s_1,l_1)_{(a)}\leq \con\, \int_{(a)}|f(s_2,l_2(\rho/2r,\vph))|d\rho^2 d\vph ds_2\,.
\end{equation}
The rhs is bounded and for $r\to \infty$ tends to $\con\int (\Delta s)^2|f(s_2,l_1)|ds_2$. In region (b) there is $2r\geq\rho\geq|\Delta s|$ and we find
\begin{equation}
 I(s_1,l_1)_{(b)}\leq \con\, \int_{(b)}\frac{|f(s_2,l_2(\rho/2r,\vph))|}{(\la+\rho-|\Delta s|)^\kappa}d\rho^2 d\vph ds_2\,.
\end{equation}
In case (i) the rhs is bounded by $\con(\la+r)^{2-\kappa}$, which closes the proof of this case. In case (ii) the limit of the rhs for $r\to \infty$ is bounded by \mbox{$\con\int (\la+|\Delta s|)|f(s_2,l_1)|ds_2$}, which closes the proof of~(ii) for $B^k_\pm[r,f]$.\linebreak Statement on $B^k_\pm[g^\eta_R,f]$ is obvious.\\
Case (iii). The separation of supports is reflected in the restriction $\xi\geq\xi_0$ for some $\xi_0>0$. Here we split integration into regions (a) $|\Delta s+\Delta r|\leq[2r^2\xi^2+(\Delta r)^2]^{1/2}$ and (b) the rest. In region (a) we have
\begin{multline}
 [4r^2\xi^2+(\Delta r)^2]^{1/2}-|\Delta s+\Delta r|\geq[4r^2\xi^2+(\Delta r)^2]^{1/2}-[2r^2\xi^2+(\Delta r)^2]^{1/2}\\
 \geq\frac{r^2\xi^2}{\sqrt{4r^2\xi^2+(\Delta r)^2}}
 \geq\frac{r\xi}{\sqrt{4+(d/\xi)^2}}\geq br\xi\,,
\end{multline}
where $b=[4+(d/\xi_0)^2]^{-1/2}$ and in the second to last step we used the bound $|\Delta r|\leq dr$. Thus in this region the function $D_\kappa$ in the integrand of $I$ is bounded by $\con(\la+br\xi)^{-\kappa}$. In region (b) we have
\be
 |\Delta s|\geq[2r^2\xi^2+(\Delta r)^2]^{1/2}-|\Delta r|\geq\frac{r^2\xi^2}{\sqrt{2r^2\xi^2+(\Delta r)^2}}\geq br\xi\,.
\ee
Therefore, in this region integration over $ds_1ds_2$  gives a term bounded by $\con(\la+br\xi)^{-\ep}$ (put $D_\kappa\leq1$ and use Lemma \ref{lem:bound_convolution} in Appendix \ref{sec:convolution}). Thus we obtain
\begin{equation}
 I\leq \con\int_{\xi_0}^1\frac{r^2d\xi^2}{(\la+br\xi)^\beta}\leq\frac{\con}{b^2(\la+b\xi_0r)^{\beta-2}}\,,
\end{equation}
with $\beta$ given in the thesis.
\end{proof}

\section{Spectral properties}\label{foutr}

In this section we shall obtain a general relation between the asymptotic behavior of the operators $B^k_\pm[r,f]G_\pm(P^0)$ and the spectral properties of the Fourier transform $\wch{B^k_\pm}(p)G_\pm(P^0)$ in the neighborhood of the lightcone. We recall that all conventions on Fourier transforms are summarized in Introduction; in particular, $\ti{f}(\w,l)$ below is the $1$-dimensional transform of $f(s,l)$ in $s$, as defined there.

Our main tool in this section will be a partial result of Proposition 9 in \cite{he14'}. We rewrite the necessary result with the use of the norms $\|.\|_{p,1}$ defined in Appendix \ref{spacesp1} and ask the reader to find out their properties there.
\begin{pr}[\cite{he14'}\label{Bfibound}\footnote{Proposition 9 in \cite{he14'} assumes $[B,B^*]$ of $\kappa^\infty$ type, but in fact $\kappa_t^\infty$ is sufficient for the proof in the case of $B^k_\pm$.}]
Let $[B,B^*]$ be of $\kappa_t^\infty$-type, $\kappa\in(0,3)$, and let $k>(\kappa+1)/2$. Then the following bound holds
\begin{equation}\label{Bfi}
 \|B^k_\pm(\vph)G_\pm(P^0)\|\leq\con\|\vph\|_{p,1}\,,
\end{equation}
with $p=6/(6-\kappa)$ and the bounding constant depending on $B$, $\kappa$, $k$, and $G_\pm$.
\end{pr}
\noindent
Thus $B^k_\pm(\vph)G_\pm(P^0)$, originally defined for integrable $\vph$, extends also to the space~$L^{p,1}$, as defined in Appendix \ref{spacesp1}.

Let $f(s,l)\in\Sc_\ep$. Then it is easy to show that
\begin{equation}
 B^k_\pm[r,f]=B^k_\pm(\nu_r)\,,
\end{equation}
where $d\nu_r(x)=(2\pi r)^{-1}\delta(|\vx|-r)f(x^0-r,t+\hx)\,dx$ (recall notation introduced in Introduction). Three-dimensional Fourier transform of this measure is a smooth function
\begin{equation}\label{FnuW}
 (\F_3\nu_r)(x^0,\vp)=\frac{r}{(2\pi)^2}\int e^{-ir\vp\cdot\vl}f(x^0-r,l)\,d\W_t(l)\,.
\end{equation}
Using the expansion \eqref{expS} given in Appendix~\ref{regwave} one finds
\begin{multline}\label{Fnu}
 (\F_3\nu_r)(x^0,\vp)=\frac{i}{2\pi|\vp|}\big[e^{-ir|\vp|}f(x^0-r,t+\hp)-e^{+ir|\vp|}f(x^0-r,t-\hp)\big]\\
 +\frac{r}{(2\pi)^2}(\F_3R)(x^0-r,r\vp)\,.
\end{multline}
The inverse Fourier $3$-transform of the rest in the second line above is
\begin{equation}
 j_r(x)=\frac{1}{(2\pi r)^2}R(x^0-r,\vx/r)\,,\quad\text{so}\quad \|j_r\|_{p,1}=\frac{r^{(3/p)-2}}{2\pi}\int\|R(x^0,.)\|_pdx^0\,.
\end{equation}
Suppose that $f\in\Sc_\ep$, $\ep>2$. Then using the bound \eqref{Rzest} in Appendix~\ref{regwave} one finds that the integral on the rhs above is finite for $p>3/2$. Therefore, for such $p$ we have
\be
 \|j_r\|_{p,1}=\con\, r^{-(2-3/p)}\to0\quad \text{for}\quad r\to\infty.
\ee

Let $[B,B^*]$ be of $\kappa^\infty_t$-type with $\kappa>2$ and $k>3/2$. Then one can always find $\kappa'\in(2,3)$ such that $\kappa'\leq\kappa$ and $k>(\kappa'+1)/2$. Then by Proposition \ref{Bfibound}
\be
 \|B^k_\pm(j_r)G_\pm(P^0)\|\leq\con\,\|j_r\|_{6/(6-\kappa'),1}=\con\,r^{-(\kappa'-2)/2}\,.
\ee

As $B^k_\pm(\vph)G(P^0_\pm)$ extends both to $\vph=\nu_r$ as well as $\vph=j_r$ in the assumed case, the operator
\be
 B^k_\pm(\nu_r{-}j_r)G(P^0_\pm)=\wch{B^k_\pm}(\wh{\nu_r{-}j_r})G(P^0_\pm)
\ee
is well defined, the extension on the rhs defined in terms of the lhs. Taking time-transform of Eq.\,\eqref{Fnu} we find
\be
 \wh{(\nu_r{-}j_r)}(p)=\frac{i}{2\pi|\vp|}\Big[e^{irp^-}\ti{f}(p^0,t+\hp)-e^{irp^+}\ti{f}(p^0,t-\hp)\Big]\,.
\ee
Introducing notation
\begin{equation}\label{ppm}
 p^\pm=p^0\pm|\vp|\,,\qquad \hp^\pm=t\pm\hp\ \ (\in\lct)\,,
\end{equation}
which we shall use from now on, we summarize and extend the result of our preceding discussion as follows.
\begin{pr}\label{asmom} Let $\kappa>2$, $k>3/2$, and let $f\in\Sc_\ep$, $\ep>2$.\\
(i) If $[B,B^*]$ is of $\kappa_t^\infty$-type, then
\begin{multline}\label{Brp}
 B^k_\pm[r,f]G_\pm(P^0)\\
 =\frac{i}{2\pi}\int\Big[e^{irp^-}\ti{f}(p^0,\hp^+)-e^{irp^+}\ti{f}(p^0,\hp^-)\Big]\wch{B^k_\pm}(p)G_\pm(P^0)|\vp|^{-1}dp +O_{\|.\|}(r^{-\beta})\,.
\end{multline}
(ii) If in addition $[\p_aB,\p_bB^*]$ is of $\kappa^\infty_t$-type, then
\begin{multline}\label{dBrp}
 (\p_aB)^k_\pm[r,l^af]G_\pm(P^0)\\
 =\frac{-1}{2\pi}\int\Big[e^{irp^-}p^-\ti{f}(p^0,\hp^+)
 -e^{irp^+}p^+\ti{f}(p^0,\hp^-)\Big]\wch{B^k_\pm}(p)G_\pm(P^0)|\vp|^{-1}dp\\
 +O_{\|.\|}(r^{-\beta})\,,
\end{multline}
\begin{multline}\label{BBrp}
 B^k_\pm[r,\df]G_\pm(P^0)+(\p_aB)^k_\pm[r,l^af]G_\pm(P^0)\\
 =\frac{1}{2\pi}\int\Big[e^{irp^-}\ti{f}(p^0,\hp^+)+e^{irp^+}\ti{f}(p^0,\hp^-)\Big]\wch{B^k_\pm}(p)G_\pm(P^0)dp
 +O_{\|.\|}(r^{-\beta})\,,
\end{multline}
where $\beta>0$ and  $O_{\|.\|}$~indicates a bound in norm. Both sides of these equalities are bounded in norm uniformly with respect to $r$.
\end{pr}
\begin{proof}
By Theorem~\ref{limit} the lhs of \eqref{Brp} is bounded in norm uniformly with respect to $r$, and the formula itself summarizes the result of the preceding discussion. Relation \eqref{dBrp} follows directly from~(i): one substitutes on the rhs of Eq.\,\eqref{Brp}
\be\label{subst}
\ti{f}(p^0,\hp^\pm)\rightarrow (\hp^\pm)^a\ti{f}(p^0,\hp^\pm)\,,\quad \wch{B}(p)\rightarrow ip_a\wch{B}(p)\,,
\ee
and notes that $\hp^\pm\cdot p=p^\mp $. Relation \eqref{BBrp} follows from the former two formulas.
\end{proof}
We note the following identity for further use:
\be\label{derB}
 \p_rB^k_\pm[r,f]=(\p_aB)^k_\pm[r,l^af]+r^{-1}B^k_\pm[r,f]\,,
\ee
which follows from the Definition~\ref{smeared} for $B\in\C^1$.

The results of Proposition~\ref{asmom} are strengthened by smearing in $r$. We recall that $w$ is a function of $R$ as given before Eq.\,\eqref{geta}, and also note that
\be
 \ti{g^\eta_R}(u)=e^{i(R-w)u}\ti{g}(wu)\,.
\ee
\begin{pr}\label{asmomg}
 Let $\kappa>2$, $k>3/2$, and let $f\in\Sc_\ep$, $\ep>2$.\\
(i) If $[B,B^*]$ is of $\kappa_t^\infty$-type, then
\begin{multline}\label{Brpg}
 B^k_\pm[g_R,f]G_\pm(P^0)
 =\pm i\int \ti{g}(Rp^\mp)\ti{f}(p^0,\hp^\pm)\wch{B^k_\pm}(p)G_\pm(P^0)|\vp|^{-1}dp\\
 +O_{\|.\|}(R^{-\gamma_1})\,,
\end{multline}
(ii) If in addition $[\p_aB,\p_bB^*]$ is of $\kappa^\infty_t$-type, then
\begin{equation}\label{dBrpg}
 (\p_aB)^k_\pm[g_R^\eta,l^af]G_\pm(P^0)=O_{\|.\|}(w^{-1})\,,
\end{equation}
\begin{multline}\label{BBrpgg}
 B^k_\pm[g_R^\eta,\df]G_\pm(P^0)
 =\int\Big[\ti{g_R^\eta}(p^-)\ti{f}(p^0,\hp^+)+\ti{g_R^\eta}(p^+)\ti{f}(p^0,\hp^-)\Big]\wch{B^k_\pm}(p)G_\pm(P^0)dp\\ +O_{\|.\|}(R^{-\gamma_2})\,,
\end{multline}
\begin{equation}\label{BBrpg}
 B^k_\pm[g_R,\df]G_\pm(P^0)=\int \ti{g}(Rp^\mp)\ti{f}(p^0,\hp^\pm)\wch{B^k_\pm}(p)G_\pm(P^0)dp +O_{\|.\|}(R^{-\gamma_3})\,,
\end{equation}
where $\gamma_i$ are some positive numbers. Both sides of these equalities are bounded in norm uniformly with respect to $R$.
\end{pr}
\begin{proof}
(i) Smearing relation \eqref{Brp} with $g_R$ one obtains formula which differs from \eqref{Brpg} by the `wrong' term
\begin{equation}\label{wrong}
 \mp i\int \ti{g}(Rp^\pm)\ti{f}(p^0,\hp^\mp)\wch{B^k_\pm}(p)G_\pm(P^0)|\vp|^{-1}dp
\end{equation}
on the rhs. We choose numbers $\kappa'\in(2,3)$, $\kappa'<\kappa$, and $\delta>0$ such that
\be\label{numbers}
 (\kappa'+1)/2+\delta<k\,.
\ee
We write the term \eqref{wrong} as $\mp ie^{\mp i\delta\pi/2}B^{k-\delta}_\pm(\chi^\pm_R)$, where $\chi^\pm_R$ is the inverse transform of
\be
 \wh{\chi^\pm_R}(p)=\theta(\pm p^0)|p^0|^\delta|\vp|^{-1} \ti{g}(Rp^\pm)\ti{f}(p^0,\hp^\mp)\,.
\ee
We note that $k-\delta>(\kappa'+1)/2$, so the use of Proposition~\ref{Bfibound} gives then the first estimate in the following sequence:
\be\label{chiest}
 \|B^{k-\delta}_\pm(\chi_R^\pm)\|\leq\con\,\|\chi_R^\pm\|_{q,1}\leq \con\,R^{-(\kappa'-2)/2-\delta}\,,
\ee
with $q=6/(6-\kappa')$, and the proof of the second inequality is a~more technical point, which we shift to Appendix~\ref{nest}. \\
(ii) We smear identity \eqref{derB} with $g_R^\eta(r)$. The term on the lhs may be integrated by parts, which yields $-w^{-1}B[g'{}_{\!R}^\eta,f]$. As the second term on the rhs is $O_{\|.\|}(R^{-1})$ upon smearing, the relation \eqref{dBrpg} follows. Next, we smear identity~\eqref{BBrp} with $g_R^\eta(r)$, which yields \eqref{BBrpgg}. For $\eta=1$ the smeared `wrong' term on the rhs is estimated by $R^{-\gamma_3}$ with the use of the method applied in (i) and one obtains \eqref{BBrpg}.
\end{proof}

We end this section with a remark on further smearing of operators of the form $[B^k_\pm(\nu_r{-}j_r)G(P^0_\pm)](\vx)$ with an integrable function $h(\vx)$, which will be needed in the next section. As a result one obtains $B^k_\pm(\nu_r{-}j_r)(h)G(P^0_\pm)$. However, we note that $\nu_r{\stackrel{(3)}{*}}h$ is integrable (with {${\stackrel{(3)}{*}}$} denoting the $3$-space convolution), while $j_r{\stackrel{(3)}{*}}h$ is in $L^{p,1}$ (due to the relation \eqref{young13}). Therefore the smeared operator may be further written as
\be\label{3fourierp1}
 B^k_\pm(\nu_r{-}j_r)(h)G(P^0_\pm)=B^k_\pm((\nu_r{-}j_r){\stackrel{(3)}{*}}h)G(P^0_\pm)=2\pi\wch{B^k_\pm}(\wh{(\nu_r{-}j_r)}\F_3h)G(P^0_\pm)\,.
\ee

\section{Asymptotic fields}\label{asfields}

In this section we shall find that if the asymptotic limit for $r\to\infty$ of the operators $B^k_\pm[r,f]G_\pm(P^0)$ exists, it defines a massless field.
Our investigation is thus based on the following supposition.
\begin{ass}\label{kout}
 Let $k>3/2$ and $[B,B^*]$ be of $\kappa^\infty$-type with $\kappa>2$. We~assume that for all $f\in\Sc_\ep$ with $\ep>2$ there exist weak limits
\begin{equation}\label{weak}
 B^{k\,\out}_\pm[f]G_\pm(P^0)=\mathrm{w}\!-\!\lim_{r\to\infty}B^k_\pm[r,f]G_\pm(P^0)\,.
\end{equation}
\end{ass}
\begin{dfn}\label{out}
 Let $f(s,l)=b^{(n)}(s,l)\equiv\p^nb(s,l)/\p s^n$, with $n\geq2$, $b\in\Sc_\ep$, $\ep>2$. Then we denote
\be
 B^\out_\pm[f]G_\pm(P^0)=B^{n\,\out}_\pm[b]G_\pm(P^0)\,.
\ee
\end{dfn}
\noindent
We note that for positive integer $m$ there is $B^{k+m}_\pm[r,b]=B^k_\pm[r,b^{(m)}]$, so the definition is consistent.
\begin{thm}
 Let the terms of Assumption~\ref{kout} and Definition~\ref{out} be satisfied. Then
\begin{gather}
 (\p_aB)^\out_\pm[l^af]G_\pm(P^0)=0\,,\label{plBout}\\[2ex]
 B^\out_\pm[\df]G_\pm(P^0)
 =\mathrm{w}\!-\!\lim_{r\to\infty}\frac{1}{2\pi}\int e^{irp^\mp}\ti{f}(p^0,\hp^\pm)\wch{B}(p)G_\pm(P^0)\theta(\pm p^0)dp\,,\label{Bout}\\[2ex]
 \Box\, B^\out_\pm[\df]G_\pm(P^0)=0\,.\label{Bwave}
\end{gather}
\end{thm}
\begin{proof}
The operator on the lhs of \eqref{plBout} exists as a weak limit of $(\p_aB)^n_\pm[r,l^ab]$. But according to \eqref{dBrpg} after smearing this tends to zero, so \eqref{plBout} follows.

Relation \eqref{BBrp} gives now for $k=n$ and with $b$ replacing $f$
\begin{multline}\label{Boutpm}
 B^{n\,\out}_\pm[r,\dot{b}]G_\pm(P^0)\\
 =\mathrm{w}\!-\!\lim_{r\to\infty}\frac{1}{2\pi}\int\Big[e^{irp^-}\ti{b}(p^0,\hp^+)
 +e^{irp^+}\ti{b}(p^0,\hp^-)\Big]\wch{B^n_\pm}(p)G_\pm(P^0)dp\,.
\end{multline}
We need to show that the two terms on the rhs of Eq.\,\eqref{Boutpm} have separate limits. To this end we denote
\be
 (\F_3\vec{h})(\vp)=(2\pi)^{-1}e^{\mp i\delta\pi/2}\hp\,|\vp|^\delta(1+\la^2|\vp|^2)^{-2}\,,
\ee
with \mbox{$\delta\in(0,1)$}, $\delta<n-(3/2)$. This function satisfies the assumptions of Lemma~14 in \cite{he14'} with $\gamma=\delta$, so $\vec{h}(\vx)$ is integrable. Consider relation \eqref{BBrp} with the replacements
\be
 B\to C=(1-\la^2\Delta)^2B\,,\quad k\to n-\delta\,,\quad f(s,l)\to\vec{l}\,b(s,l)\,.
\ee
As $n-\delta>3/2$ the lhs still has a finite week limit for $r\to\infty$ (by Assumption~\ref{kout}). Take the translation of the resulting formula by $\vx$, take the scalar product with $\vec{h}(\vx)$ and integrate over $d^3x$. As $\vec{h}$ is integrable, the lhs again has a finite weak limit for $r\to\infty$. Also, by integrability of $\vec{h}$ the contribution of the rest vanishes in the limit. We turn to the integral on the rhs of \eqref{BBrp}. Taking into account the remark closing Section \ref{foutr} the effect of the described operations obeys the scheme \eqref{3fourierp1} and amounts to the replacements:
\be
 f(p^0,\hp^\pm)\to 2\pi(\F_3\vec{h})(\vp)\cdot(\pm\hp) b(p^0,\hp^\pm)\,,\quad \wch{B^k_\pm}(p)\to (1+\la^2|\vp|^2)\wch{B^{k-\delta}_\pm}(p)\,.
\ee
Substituting here the transform $(\F_3\vec{h})(\vp)$, using $\hp\cdot\hp=-1$ (Lorentz product) and recalling the definition \eqref{Bkdef} we obtain expression similar to the rhs of \eqref{Boutpm}, but with minus sign in front of the first term in brackets. Thus separate limits of the two terms exist. But from the proof of Proposition~\ref{asmomg} (ii) we know, that after smearing with $g_R(r)$ the `wrong' term vanishes in the limit, so the limit before smearing is also zero. Transferring the derivatives from $B$ to $b$ we arrive at the relation~\eqref{Bout}.

To prove \eqref{Bwave} we note that acting with $\Box$ on \eqref{Bout} one produces on the rhs additional factor $-p^2$ under the integral. On the other hand, using \eqref{Bout} we also find:
\begin{equation}
 (\p_aB)^\out_\pm[l^a\df]G_\pm(P^0)
 =\mathrm{w}\!-\!\lim_{r\to\infty}\frac{i}{2\pi}\int e^{irp^\mp}p^\mp\ti{f}(p^0,\hp^\pm)\wch{B}(p)G_\pm(P^0)\theta(\pm p^0)dp\,;
\end{equation}
this is shown as in the proof of Proposition \ref{asmom}, by replacements \eqref{subst} applied to the rhs of \eqref{Bout}.
Therefore, noting that $-p^2=(p^0\mp|\vp|)^2-2p^0(p^0\mp|\vp|)$, we can write
\be
 \Box\, B^\out_\pm[\df]G_\pm(P^0)=-(\p_a\p_bB)^\out_\pm[l^al^b\df]G_\pm(P^0) - 2(\p_aB)^\out_\pm[l^af^{(2)}]G_\pm(P^0)\,.
\ee
But this vanishes according to \eqref{plBout}.
\end{proof}

We now turn to the question of the (in)dependence of the asymptotic fields on the choice of time axis. To address this question we shall assume that our  Assumption~\ref{kout} is valid in all frames -- we recall that up to now the time axis vector $t$ was kept fixed.

We first rewrite definition \eqref{smearf} in a way free from the assumption on particular gauge of vectors $l$. We note that if $F(s,l)$ is a function homogeneous of degree $-m$, then the integral $\int F(s,l)ds$ is a~function of $l$ homogeneous of degree $-m+1$. Let $\df(s,l)$ be homogeneous function of degree $-2$. Then Definition~\ref{smearf} (where we use $\df$ in place of $f$) may be written as
\be\label{smearhom}
 B[r,\df]=\frac{1}{2\pi}\int \frac{r}{t\cdot l}B\Big(\frac{st+rl}{t\cdot l}\Big)\df(s,l)\,ds\,d^2l\,;
\ee
clearly in $t$-gauge this reduces to the former form, but has the advantage that the dependence on vector $t$ is now explicit. We differentiate the $t$-dependent expression under the integral with respect to $t^a$ and find:
\begin{multline}
 \frac{\p}{\p t^a}\Big[\frac{r}{t\cdot l}B\Big(\frac{st+rl}{t\cdot l}\Big)\Big]\\
 =\frac{rs}{(t\cdot l)^3}\
 \big(t\cdot l\, \delta^b_a-l_at^b\big)(\p_bB)\Big(\frac{st+rl}{t\cdot l}\Big)
 -r\p_r\Big[\frac{rl_a}{(t\cdot l)^2}B\Big(\frac{st+rl}{t\cdot l}\Big)\Big]\,.
\end{multline}
Thus restoring on the rhs the $t$-gauge we can write
\be
 (\p/\p t^a)B[r,\df]=(\p_bB)[r,(\delta_a^b-l_at^b)s\df]-r\p_rB[r,l_af]\,.
\ee
Smearing this with $g_R(r)$ and integrating the far right element by parts we obtain
\be\label{dtg}
 (\p/\p t^a)B[g_R,\df]=(\p_bB)[g_R,(\delta_a^b-l_at^b)s\df]+B[h_R,l_a\df]\,,
\ee
where $h_R(r)=R^{-1}h(R^{-1}r)$, $h(u)=d[ug(u)]/du$.
\begin{thm}
 Let the terms of Assumption~\ref{kout} be satisfied with respect to all vectors $t$, and let $f(s,l)=b^{(3)}(s,l)$, with $b\in\Sc^2_\ep$, $\ep>2$. Then operators $B^\out_\pm[\df]$ do not depend on the choice of time axis used for their definition.
\end{thm}
\begin{proof}
We have $s\,\df(s,l)=b_1^{(3)}(s,l)$, with $b_1(s,l)=s\dot{b}(s,l)-3b(s,l)\in\Sc^2_\ep$.
Calculating in $t$-gauge one finds that $B^3_\pm[r,\dot{b}]=B[r,\dot{b}_\pm^{(3)}]$, where $\ti{b_\pm}(\w,l)=\theta(\pm\w)\ti{b}(\w,l)$, and similarly for other terms in Eq.\,\eqref{dtg}. Therefore,
\be
 (\p/\p t^a)B^3_\pm[g_R,\dot{b}]=(\p_bB)^2_\pm[g_R,(\delta_a^b-l_at^b)\dot{b}_1]+B^3_\pm[h_R,l_a\dot{b}]\,.
\ee
The limits of the operators on the rhs, superposed with $G_\pm(P^0)$, exist before smearing (for $r\to\infty$). But the second operator is smeared with the derivative of a function vanishing at the end points, so its limit for $R\to\infty$ is zero. Next, we note that using \eqref{plBout} we can write the limit of the first term as:
\begin{multline}
 \mathrm{w}\!-\!\lim_{r\to\infty}(\p_bB)^2_\pm[r,(\delta_a^b-l_at^b)\dot{b}_1]G_\pm(P^0)\\
 =\mathrm{w}\!-\!\lim_{r\to\infty}(\p_bB)^2_\pm\big[r,\big(\delta_a^b-t_at^b+(l_a-t_a)(l^b-t^b)\big)\dot{b}_1\big]G_\pm(P^0)\,.
\end{multline}
The limit on the rhs is obtained with the use of \eqref{Bout}, where one has to substitute
\[
 \wch{B}(p)\to ip_b\wch{B}(p)\quad \text{and}\quad \ti{f}(p^0,\hp^\pm)\to(\delta^b_a-t_at^b+\hp_a\hp^b)\ti{b_1}(p^0,\hp^\pm)\,.
\]
But the contraction of these terms vanishes, so the limit of $B^3_\pm[g_R,\dot{b}]$ does not depend on $t$. As this limit exists before smearing with $g_R$, the thesis follows.
\end{proof}

We shall now want to be able to compose asymptotic fields. With that in mind we make a stronger supposition.
\begin{ass}\label{koutstrong}
 Let the terms of Assumption~\ref{kout} and Definition~\ref{out} be satisfied. We assume that in this context the limits exist as strong limits after smearing with $g_R(r)$; i.e.\ the limit in \eqref{weak} is equal to
\begin{equation}\label{strong}
 \mathrm{s}\!-\!\lim_{R\to\infty} B^k_\pm[g_R,f]G_\pm(P^0)\,.
\end{equation}
\end{ass}

From now on we choose $G_\pm(P^0)=(1+\la P^0)^{-1}$ and note that if $B\in\C^\infty_t$, then
\begin{equation}\label{PBP}
 (1+\la P^0)^m B(1+\la P^0)^{-m-1}=\sum_{l=0}^m\tbinom{m}{l}(-i\la\p_0)^lB\,(1+\la P^0)^{-l-1}\,;
\end{equation}
this is easily obtained by writing
\[
 \exp{[i\tau(1+\la P^0)]}\,B=B(\tau\la)\exp{[i\tau(1+\la P^0)]}
\]
and comparing $\tau^m$-terms. Therefore, under the conditions of Assumption~\ref{koutstrong} the operators $(1+\la P^0)^m B^k_\pm[g_R,f](1+\la P^0)^{-m-1}$ are bounded and have strong limits.
\begin{pr}\label{outBB}
Let $B$, $B_i$ ($i=1,\ldots,n$) satisfy the conditions of Assumption~\ref{koutstrong}. Then
\begin{gather}
 B^{k\out}_\pm[f] (1+\la P^0)^{-m-1}\Hc\subseteq (1+\la P^0)^{-m}\Hc\,,\label{mdom}\\[1ex]
\begin{aligned}
 \mathrm{s}\!-\!\lim_{R\to\infty}(1+\la P^0)^m B^k_\pm&[g_R,f](1+\la P^0)^{-m-1}\\
 &=(1+\la P^0)^m B^{k\out}_\pm[f](1+\la P^0)^{-m-1}\,,\label{mmlim}
\end{aligned}\\[1ex]
\begin{aligned}
 \mathrm{s}\!-\!\lim_{R\to\infty}B^k_{1\pm}[g_R,f_1]&\ldots B^k_{n\pm}[g_R,f_n](1+\la P^0)^{-n-1}\\
 &=B^{k\out}_{1\pm}[f_1]\ldots B^{k\out}_{n\pm}[f_n](1+\la P^0)^{-n-1}\label{nBn}
\end{aligned}
\end{gather}
(uncorrelated signs).
\end{pr}
\begin{proof}
 For any $\psi\in\Hc$ there is
\begin{gather*}
 \mathrm{s}\!-\!\lim_{R\to\infty}B^k_\pm[g_R,f](1+\la P^0)^{-m-1}\psi=B^{k\out}_\pm[g_R,f](1+\la P^0)^{-m-1}\psi\,,\\
 \mathrm{s}\!-\!\lim_{R\to\infty}(1+\la P^0)^mB^k_\pm[g_R,f](1+\la P^0)^{-m-1}\psi=\varphi
\end{gather*}
for some $\varphi\in\Hc$, the second relation by remarks preceding the proposition. But as $(1+\la P^0)^m$ is selfadjoint, its graph is closed and equations \eqref{mdom} and \eqref{mmlim} follow. Eq.\,\eqref{nBn} is an immediate consequence.
\end{proof}

\begin{thm}\label{outcom}
Let $f_i(s,l)=b_i^{(3)}(s,l)$, $b_i\in\Sc_\ep$ ($i=1,2,3$), $\ep>2$ and let the supports of the functions
\be
 \lct\ni l\mapsto\|f_i(.,l)\|_\infty=\sup_{s\in\mR}|f_i(s,l)|\,,\quad i=1,2
\ee
be disjoint. If $B^\out_\pm[f_i](1+\la P^0)^{-1}$ exist as strong limits according to Assumption~\ref{koutstrong}, then
\begin{gather}
 \big[B^\out_\pm[f_1],B^\out_\pm[f_2]\big](1+\la P^0)^{-2}=0\,,\label{comtwo}\\
 \Big[B^\out_\pm[f_1],\big[B^\out_\pm[f_2],B^\out_\pm[f_3]\big]\Big](1+\la P^0)^{-3}=0\,\label{comthree}
\end{gather}
(uncorrelated signs).
\end{thm}
\begin{proof}
Eq.\,\eqref{comtwo} is an immediate consequence of relation \eqref{nBn} and Theorem~\ref{limit} (iii). For Eq.\,\eqref{comthree} one needs in addition to decompose $f_3=f_{31}+f_{32}$, with $\supp f_i\cap\supp f_{3i}=\varnothing$ ($i=1,2$). For contribution of $f_{32}$ the equality follows directly, and for contribution of $f_{31}$ -- by Jacobi's identity.
\end{proof}

\section{Haag-Ruelle case}\label{hrtheory}

In this section we supplement our basic assumptions on the decay of commutators, Definition~\ref{com}, and on relativistic positivity of energy, by the existence of vacuum and one-particle massless states. In this framework we present the construction of scattering states which closely parallels the one used in the Haag-Ruelle scattering theory in case of massive particles (cf.\ the exposition in \cite{he14}). Neither of the Assumptions 1, 2 introduced in the preceeding section is needed in this construction. To prove the existence of the asymptotic states we suppose instead that an energy-momentum condition of the type similar to the one used in \cite{he14} is fulfilled. Similarly as in the massive case, the full Fock structure of states needs a clustering property introduced in this reference.

For completeness, we also give a short account of the case of local fields. Local fields satisfy clustering property mentioned above. Moreover, the Cook me\-thod used for the construction of asymptotic states (which needs the spectrum condition) may be replaced by an independent proof of the existence of asymptotic fields. The idea of this proof is due to Buchholz \cite{bu77}, but our construction is significantly simpler, compared to exposition in \cite{bu77}, due to the bound on the norm of the operators $B^k_\pm[g_R,f]G_\pm(P^0)$ given in Theorem~\ref{limit}~(ii). The possibility of similar simplification was already noticed by Buchholz in \cite{bu90}.

\paragraph{}
\label{hrpsi}
From now on, we consider only operators $B$ belonging to some \mbox{${}^*$-sub}\-al\-ge\-bra of $\B(\Hc)$ which we denote by $\F$ and interpret as an algebra of quantum fields. By assumption, operators from $\F$ are infinitely differentiable and their commutators are (pairwise) of $\kappa^\infty$ type with $\kappa>2$. Moreover, we assume throughout that $f=b^{(3)}$ for some $b\in\Sdeg{1}_\ep$, $\ep>2$. We recall that function~$g$ satisfies \eqref{gsupint}, which implies $\ti{g}(0)=\frac{1}{2\pi}$. Also, we put $G_\pm(E) = G(E)=(1+\lambda E)^{-1}$. Since $b\in\Sdeg{1}_\ep$, the function $f$ is homogeneous of degree~$-2$:
\begin{equation}\label{eq:hom}
  f(\mu s,\mu l)=\mu^{-2}f(s,l)\,,\qquad \ti{f}(\mu^{-1}\w,\mu l)=\mu^{-1} \ti{f}(\w,l)
\end{equation}
for all $\mu>0$. We recall notation introduced in Eq.\,\eqref{ppm}, which will also be used for energy-momentum \emph{operators}:
\begin{equation}
 P^\pm=P^0\pm |\vec{P}|\,,\qquad \hat{P}^\pm=t\pm\hat{P}\,.
\end{equation}

Let us define the following operators:
\begin{align}
 &B^\pm_R[f](1+\lambda P^0)^{-1}\phantom{'}= B^3_\pm[g_R,b](1+\lambda P^0)^{-1},\label{eq:hr_B_def}\\
 &B'^\pm_R[f](1+\lambda P^0)^{-1}=\pm i \int \ti{g}(Rp^\mp)\ti{f}(\pm 1,|p^0|\hp^\pm)\wch{B}(p)(1+\lambda P^0)^{-1}\theta(\pm p^0)dp.\label{eq:hr_B'_def}
\end{align}
Taking into account that $\ti{f}(p^0,\hp^\pm)/p^0=\sgn p^0\, \ti{f}(\sgn p^0,|p^0|\hp^\pm)$ (by Eq.\,\eqref{eq:hom}) and using the result of Theorem~\ref{asmomg}~(ii), equation \eqref{BBrpg}, we obtain
\begin{equation}\label{eq:diff_B_and_B'}
 \|(B^\pm_R[f] - B'^\pm_R[f])(1+\lambda E)^{-1}\|=O(R^{-\gamma})
\end{equation}
for some $\gamma>0$. This estimate will allow us to use operators $B^\pm_R[f]$ and $B'^\pm_R[f]$ interchangeably. In particular, all bounds listed in this and the following paragraph apply to both of these operators.

As a result of Theorem~\ref{limit}~(ii) the norm
\begin{equation}
 \|B^\pm_{R}[f](1+\la P^0)^{-1}\|=\|B^3_\pm[g_R,b](1+\la P^0)^{-1}\|
\end{equation}
is bounded by a constant independent of $R$. Using equation \eqref{PBP} we show that
\begin{equation}
 \|(1+\la P^0)^{m} B^\pm_{R}[f] (1+\la P^0)^{-1-m}\| \leq \con_{m}
\end{equation}
from which it follows that
\begin{equation}\label{eq:norm_G_many}
 \|(1+\la P^0)^{m} B^\pm_{1,R}[f_1]\ldots B^\pm_{n,R}[f_n] (1+\la P^0)^{-n-m}\| \leq \con_{m,n}
\end{equation}
for $m,n\in\mN$.

\paragraph{}
\label{hrcom}
Let the supports of $\|f_i(\cdot,l)\|_\infty=\sup_{s\in\mR}|f_i(s,l)|$ ($i=1,2$) be disjoint. Then:
\begin{gather}
 \big\|(1+\la P^0)^{m}\big[B^\pm_{1,R}[f_1],B^\pm_{2,R}[f_2]\big](1+\la P^0)^{-2-m}\big\|=O(R^{-\gamma}),\label{eq:bound_com_2}\\[1ex]
 \big\|(1+\la P^0)^{m}\big[\tfrac{d}{dR} B^\pm_{1,R}[f_1],B^\pm_{2,R}[f_2] \big] (1+\la P^0)^{-2-m}\big\|=O(R^{-(1+\gamma)}),\label{eq:bound_com_2d}\\[1ex]
 \big\|(1+\la P^0)^{m}\big[B^\pm_{1,R}[f_1],\big[B^\pm_{2,R}[f_2],B^\pm_{3,R}[f_3]\big]\big](1+\la P^0)^{-3-m}\big\|=O(R^{-\gamma})\,\,\label{eq:bound_com_3}
\end{gather}
(uncorrelated signs) for some $\gamma>0$. For $m=0$,  the first bound follows directly from Theorem~\ref{limit}, the second bound is the consequence of the first one since
\begin{equation}\label{eq:R_derivative_B}
 \frac{d}{dR} B^\pm_{R}[f] = \frac{1}{R} B^3_\pm[h_R,b]\,,
\end{equation}
where $h_R(r)=\frac{1}{R}h(r/R)$, $h(r)=-g(r)- r g'(r)$. The third bound may be obtained by the method used in the proof of Theorem~\ref{outcom}. To generalize the above result for any $m\in\mN$ we use the the identity \eqref{PBP}.

\paragraph{}
\label{hrspec}
Now one assumes the existence of the vacuum - the unique up to a phase, unit vector $\W\in\Hc$ which is invariant under the action of translation operators $U(x)$. Let $E(A)$, $A\subseteq\hM$,  be the spectral family of the four-momentum operators. Then for $\mu\geq0$ we denote
\begin{equation}\label{eq:def_E_mu}
  E_\mu=E\big(\{p\mid 0\leq p^2\leq\mu^2\,,\ p^0\geq0 \}\big)\,.
\end{equation}
In particular $E_0 \Hc\subset \Hc$ is the subspace of massless one-particle states.

We say that the operator $B\in\F$ fulfils the (energy-momentum) spectral condition if for some $\ep>0$
\begin{equation}\label{eq:spectral_condition}
 \int_0^\ep\|(E_\mu-E_0)B\W\|\frac{d\mu}{\mu}<\infty.
\end{equation}
For such operators the following integrability condition holds
\begin{equation}\label{eq:dB_dR_bound}
 \int_0^\infty\Big\|(1+\la P^0)^{n}\frac{d B'^+_R[f]\W}{d R}\Big\|\,d R <\infty\,,
\end{equation}
which will be used to show the existence of asymptotic states. To prove this, let us note that, since $U(x)\Omega=\Omega$,
\begin{equation}\label{eq:B'_Omega}
 B'^+_R[f]\Omega= i (2\pi)^2 \ti{g}(R P^-)\ti{f}(1,P^0\hat{P}^+)B\Omega\,.
\end{equation}
Thus, because $\ti{f}(1,\omega l)=\ti{f}(\omega,l)/\omega=(-i)^3\omega^2 \ti{b}(\omega,l)$ for $\omega>0$,
\[
\frac{d B'^+_R[f]\W}{d R} =-(2\pi)^2P^-\,\ti{g}'(R P^-)(P^0)^2 \ti{b}(P^0,\hat{P}^+) B\W\equiv TR^{-1}\phi(RP^2)B\W\,,
\]
where $\ti{g}'(\w)=d\ti{g}(\w)/d\w$, $\phi(u)=u/(1+\la u)^2$ and
\[
 T=-(2\pi)^2 P^0(P^+)^{-1} \ti{g}'(R P^-)P^0\ti{b}(P^0,\hat{P}^+)(1+\la RP^2)^2\,,
\]
and we have used relation $P^2=P^-P^+$. We observe that $\|(1+\la P^0)^nT\|$ is bounded by a constant independent of $R$: this follows from the estimates
\begin{align}
 &\|P^0(P^+)^{-1}\|\leq1\,,\ \|P^+(P^0)^{-1}\|\leq2\,,\ \|(1+\la P^0)^{n+2}P^0\ti{b}(P^0,\hat{P}^+)\|\leq\con\,,\\
 &\|\ti{g}'(R P^-)(1+\la P^0)^{-2}(1+\la RP^2)^2\|\\
 &\hspace{5em}\leq\|\ti{g}'(R P^-)(1+R P^-)^2 \| \|(1+\la P^0)^{-2}(1+\la P^+)^2 \|\leq\con\,.
 \end{align}
Therefore,
\begin{multline}
 \left\|(1+\la P^0)^{n}\frac{d B'^+_R[f]\W}{d R}\right\|\leq \con \frac{1}{R} \|(\id-E_0)\phi(R P^2)B\W\|
 \\
 \leq\con \frac{1}{R} \|(E_{R^{-1/4}}-E_0) B\W\| + \con \frac{1}{R^{3/2}} \|(\id-E_{R^{-1/4}}) B\W\|\,,
\end{multline}
where we used the following facts: \mbox{$E_0\phi(RP^2)=0$}, function $\phi$ is bounded and $\sup_{u>R^{1/2}}|\phi(u)|<\con\, R^{-1/2}$. Both terms on the rhs are integrable, the first one by the spectral condition.

The spectral condition \eqref{eq:spectral_condition} is fulfilled in particular in the vacuum representation of free massless theory by the Weyl operators -- we show this in Appendix~\ref{app:spectral_condition_weyl}.\footnote{Note that in the vacuum representation of the massive free theory, the spectral condition formulated in that context \cite{he14}, \cite{dy05} is trivially fulfilled by the Weyl operators. In massless case this fact is not so evident.}

\paragraph{} 
\label{hrasst}
Let the spectral condition \eqref{eq:spectral_condition} be satisfied for all $B_j$ ($j=1,\ldots,n$). Then for $f_j$ such that $\|f_j(\cdot,l)\|_\infty$ have disjoint supports there exist limits
\begin{equation}\label{eq:lim_B_many}
 \lim_{R\to\infty}B^+_{1,R}[f_1]\ldots B^+_{n,R}[f_n]\W =\lim_{R\to\infty}B'^+_{1,R}[f_1]\ldots B'^+_{n,R}[f_n]\W\,.
\end{equation}
These limits depend only on the one-operator asymptotic vectors
\begin{multline}\label{eq:lim_B_one}
 \lim_{R\to\infty}B^+_{R}[f]\W = \lim_{R\to\infty}B'^+_{R}[f]\W
 \\
 =(2\pi)^2 i \lim_{R\to\infty} \ti{g}(RP^-) \ti{f}(1,P^0\hat{P}^+)B\W\,
 =
 2\pi i E_0  \ti{f}(1,P) B\W\,,
\end{multline}
which are independent of the choice of timelike vector $t$. The second equality follows from \eqref{eq:B'_Omega} and the third -- from the following fact:
if $\dsp\lim_{n\rightarrow\infty} f_n(x)=f(x)$ pointwise and the sequence $\|f_n\|_\infty$ is bounded, then $f_n(A)\rightarrow f(A)$ strongly for each selfadjoint operator $A$ (c.f. Theorem VIII.5 (d) in \cite{rs72}). The structure is nontrivial if, and only if, $E_0\neq0$ and there exist operators $B$ which interpolate between $\W$ and $E_0\Hc$.

To prove the existence of the limits \eqref{eq:lim_B_many} we use Cook method, i.e.\ we show that
 $\big\|\frac{d}{dR}\big(B'^+_{1,R}[f_1]\ldots B'^+_{n,R}[f_n]\W\big)\big\|$ is integrable. After commuting all operators with derivative $\frac{d}{dR}B'^+_{j,R}[f_j]$ to the right we obtain the terms of the form
\begin{multline}
 \Big\| B'^+_{1,R}[f_1]\ldots \breve{k} \ldots B'^+_{n,R}[f_n] \frac{d}{dR}B'^+_{k,R}[f_k] \W\Big\|\\
 \leq  \big\| B'^+_{1,R}[f_1]\ldots \breve{k} \ldots B'^+_{n,R}[f_n] (1+\la P^0)^{-n+1}\big\|\, \Big\| (1+\la P^0)^{n-1} \frac{d}{dR}B'^+_{k,R}[f_k] \W\Big\|\,,
\end{multline}
which are integrable by the result of last paragraph and Eq.\,\eqref{eq:norm_G_many}. The terms containing commutators are bounded by $\con\, R^{-1-\delta}$ due to the estimates  \eqref{eq:norm_G_many} and \eqref{eq:bound_com_2d}. The statement on the limits with operators without primes is shown inductively with the use of the estimate \eqref{eq:diff_B_and_B'}.

\paragraph{}
\label{hrasst_loc}

In this paragraph, following the idea due to Buchholz \cite{bu75,bu77,bu90}, we show the existence of the limits of operators $B^\pm_R[f](1+\la P^0)^{-1}$ and \mbox{$B'^\pm_R[f](1+\la P^0)^{-1}$} as $R\rightarrow\infty$ in the framework of local quantum physics \cite{ha92}. We assume that the algebra of field operators $\F$ is the global algebra of the net $\mathcal{O}\mapsto\F(\mathcal{O})$ of local algebras of fields $\F(\mathcal{O})$ localized in bounded regions of spacetime~$\mathcal{O}$. The net $\mathcal{O}\mapsto\F(\mathcal{O})$  acts irreducibly on the Hilbert space and fulfils the following axioms: (1) $B(x) \in \F(\mathcal{O}+x)$ for $B\in \F(\mathcal{O})$ (covariance), (2) $[B_1,B_2]=0$ for $B_j\in\F(\mathcal{O}_j)$ if the regions $\mathcal{O}_1$, $\mathcal{O}_2$ are spatially separated (local commutativity), (3) $\F(\mathcal{O}_1)\subset\F(\mathcal{O}_2)$ if $\mathcal{O}_1\subset\mathcal{O}_2$ (isotony). The spectral assumption formulated in paragraph~\ref{hrspec} is not needed in this paragraph.

Chose any operator $B$ localized in some relatively compact set $\mathcal{O}_B$ and function $d\in \Sc_\ep$ such that $s\mapsto\|d(s,\cdot)\|_\infty=\sup_{l\in \lct}|d(s,l)|$ has compact support. Let $a\in M$ be any point in Minkowski spacetime such that the set $\mathcal{O}_{B,d}=\mathcal{O}_B+\{s t\mid \|d(s,\cdot)\|_\infty\neq 0\}$ is contained in the past-directed lightcone with the vertex at $a$. Then, as noted by Buchholz \cite{bu75,bu77}, for any $A\in\F(a+V_+)$ and sufficiently large $r$ the localization regions of operators $A$ and $B[r,d]$ become spatially separated. Thus, for large $R$ we have $[A,B[g_R,d]] = 0$.

Using the idea described above and the bound on the norm of the operators $B^k_\pm[g_R,d](1+\lambda P^0)^{-1}$ given in Theorem~\ref{limit} (ii) we obtain asymptotic creation/annihilation operators.

\begin{thm}\label{thm:lim_op_local}
Let $B\in\F(\mathcal{O}_B)$ be operator localized in bounded region $\mathcal{O}_B$ and $f=d^{(4)}$ for some function $d\in\Sdeg{2}_\ep$, such that $\dsp s\mapsto\|d(s,\cdot)\|_\infty=\sup_{l\in \lct}|d(s,l)|$ has compact support. Then the limit $\dsp\mathrm{s}\!-\!\lim_{R\to\infty}B^\pm_R[f] (1+\lambda P^0)^{-1}$ exists.
\end{thm}
\begin{proof}
It follows from  Theorem~\ref{asmomg}~(ii) (equation \eqref{BBrpg}) that
\begin{equation}
 B^3[g_R,d](1+\lambda P^0)^{-1} = \int \ti{g}(Rp^\mp)\ti{d}(p^0,\hp^\pm)\wch{B^2_\pm}(p)(1+\lambda P^0)^{-1} dp + O_{\|.\|}(R^{-\gamma_3})\,.
\end{equation}
Using the fact that $U(x)\Omega=\Omega$ and then Theorem VIII.5 (d) in \cite{rs72} we find that
\begin{equation}
 \lim_{R\to\infty}B^3[g_R,d]\W = (2\pi)^2 \lim_{R\to\infty} \ti{g}(RP^-)\ti{d}(P^0,\hat{P}^+)B^2\W\, =
 2\pi E_0 \ti{d}(P^0,\hat{P}^+) B^2\W\,.
\end{equation}
Hence for any $A\in\F(a+V_+)$
\begin{equation}\label{eq:local_lim_B_R_A}
 \lim_{R\to\infty}B^3[g_R,d]A\W=  2\pi A E_0 \ti{d}(P^0,\hat{P}^+) B^2\W\,.
\end{equation}
Under our assumptions the set
\begin{equation}
 \textrm{span}\{A \W\in\Hc, A\in \F(a+V_+) \}
\end{equation}
is a dense subspace of the Hilbert space \cite{bu75}. It follows from Theorem~\ref{limit}~(ii) that $\|(1+\lambda P^0)^{-1} B^3[g_R,d]\|=\|(B^*)^3[\overline{g_R},\overline{d}](1+\lambda P^0)^{-1}\|<\infty$. This, together with \eqref{eq:local_lim_B_R_A}, implies the existence of the limit $\dsp\mathrm{s}\!-\!\lim_{R\to\infty}(1+\lambda P^0)^{-1} B^3[g_R,d]$. The existence of the limit $\dsp\mathrm{s}\!-\!\lim_{R\to\infty} B^3[g_R,d](1+\lambda P^0)^{-1}$ follows from the identity
\begin{multline}
 B^3[g_R,d](1+\lambda P^0)^{-1}
  \\ =(1+\lambda P^0)^{-1} B^3[g_R,d]+ i\lambda(1+\lambda P^0)^{-1}B^4[g_R,d](1+\lambda P^0)^{-1}\,.
\end{multline}

Since under present assumptions $f=d^{(4)}$, we have
\begin{multline}\label{eq:thm_loc_B_F}
 B^+_R[f](1+\lambda P^0)^{-1}=B^4_+[g_R,d](1+\lambda P^0)^{-1}
 \\
 =\int \phi(s)\, C^3(ts)[g_R,d](1+\lambda P^0)^{-1} ds\,,
\end{multline}
where $C=B+B^3$ and
\begin{equation}
 \ti{\phi}(\w)=\frac{\theta(\w)(-i\w)}{1+(-i\w)^3}.
\end{equation}
By Lemma~14 in Appendix A of \cite{he14'} we have $|\phi(s)|\leq \con\,(\la+|s|)^{-2}$. Thus $\phi$ is absolutely integrable. For any $s$ operator $C(ts)$ is local so, as shown above, the limit $\mathrm{s}\!-\!\lim_{R\to\infty} C^3(ts)[g_R,d](1+\lambda P^0)^{-1}$ exists and the integrand in the last line of \eqref{eq:thm_loc_B_F} converges pointwise. As $\phi$ is integrable and
\begin{equation}
 \|C^3(ts)[g_R,d](1+\lambda P^0)^{-1}\|=\|C^3[g_R,d](1+\lambda P^0)^{-1}\|\leq \con\,,
\end{equation}
the limit $\dsp\mathrm{s}\!-\!\lim_{R\to\infty} B^+_R[f](1+\lambda P^0)^{-1}$ exists by Lebesgue's theorem. Similar reasoning may be used to show the existence of $\dsp\mathrm{s}\!-\!\lim_{R\to\infty} B^-_R[f](1+\lambda P^0)^{-1}$.
\end{proof}

The above result together with Proposition~\ref{outBB} imply the existence of scattering states \eqref{eq:lim_B_many} without using the spectral condition under the following assumptions: (1) operators $B_j$ are local, (2) $f_j=d_j^{(4)}$ and (3) $\|d_j(s,\cdot)\|_\infty$ have compact supports.

\paragraph{} 
\label{hreta}
To obtain the Fock structure of asymptotic states we introduce operators $B^\pm_{R,\eta}[f]=B^3_\pm[g^\eta_R,b]$, with $g^\eta_R$ as defined by Eq.\,\eqref{geta}. Note that $g_R=g^\eta_R$ and $B_R[f]=B_{R,\eta}[f]$ for $\eta=1$. For $\eta\in(0,1\>$ there is
\begin{equation}\label{eq:lim_B_many_eta}
\lim_{R\to\infty}B^{+}_{1,R,\eta}[f_1]\ldots B^{+}_{n,R,\eta}[f_n]\W=\lim_{R\to\infty}B^+_{1,R}[f_1]\ldots B^+_{n,R}[f_n]\W.
\end{equation}
First, we show this for $n=1$. Using Theorem~\ref{asmomg}~(ii) (equation \eqref{BBrpgg}) and the invariance of the vacuum under the action of $U(x)$ we get
\begin{multline}
 B^+_{R,\eta}[f]\W = B^3_+[g^\eta_R,b] \W
 \\
 =(2\pi)^2 i \left(\ti{g^\eta_R}(P^-) \ti{f}(1,P^0\hat{P}^+) + \ti{g^\eta_R}(P^+)\ti{f}(1,P^0\hat{P}^-)\right)B\W
 + O_{\|.\|}(R^{-\beta})
\end{multline}
for some $\beta>0$. The function $\ti{g^\eta_R}(u)=\exp(i (R-w) u)\ti{g}(w u)$, \mbox{$w=\la (R/\la)^\eta$}, is bounded and converges pointwise as $R\rightarrow\infty$ to $\frac{1}{2\pi}$ for $\w=0$ and $0$ for $\w\neq 0$. Thus, using Theorem VIII.5 (d) in \cite{rs72} we get
\begin{equation}\label{eq:lim_B_one_eta}
 \lim_{R\to\infty}B^+_{R,\eta}[f]\W =2\pi i E_0  \ti{f}(1,P) B\W = \lim_{R\to\infty}B^+_R[f]\W \,.
\end{equation}
Next, we note that the bounds \eqref{eq:norm_G_many}, \eqref{eq:bound_com_2} and \eqref{eq:bound_com_3} remain valid if we replace any of the operators $B^\pm_{j,R}$ with $B^\pm_{j,R,\eta}$. Using these bounds one may easily adapt the proof of Lemma 2.4 of \cite{dy05} to show \eqref{eq:lim_B_many_eta}.

\paragraph{} 
\label{hrasum}
For the derivation of the Fock structure we need an additional assumption on clustering property of commutators. We denote by $E_\W^\bot$ the projection onto the subspace orthogonal to the vector $\W$ and introduce the function $K$:
\begin{multline}\label{4psi}
 K(x_1-x_2,x_3-x_4,\tfrac{1}{2}(x_1+x_2-x_3-x_4)) \\
 =(\W,B_{12}(x_1-x_2)E_\W^\bot  U\big(-\tfrac{1}{2}(x_1+x_2-x_3-x_4)\big)B_{34}(x_3-x_4)\W)\\
 =(\W,[B_1(x_1),B_2(x_2)] E_\W^\bot[B_3(x_3),B_4(x_4)]\W)\,,
\end{multline}
where $B_{ij}(z)=[B_i(z/2),B_j(-z/2)]$.
\begin{ass}\label{cluster}
Let $B_i\in\F$, $i=1,\ldots,4$, and $N$ be any positive integer. Then for large enough, positive $d$, and
\begin{equation}
 |y_1|\leq d\,,\quad |y_2|\leq d\,,\quad |\vec{y}|\geq|y^0|+c_1d\,,
\end{equation}
the following estimate holds
\begin{equation}\label{clusterest}
 |K(y_1,y_2,y)|\leq c_2\frac{d^M}{(|\vec{y}|-|y^0|)^\nu}+c_3d^{-N}\,,
\end{equation}
and the positive constants $c_i$, $M$ and $\nu$ do not depend on $d$.

The assumption is covariant: if it holds in any particular reference system, it is valid in all other, with some other constants $c_i$.
\end{ass}

\begin{pr}[\cite{he14}]\label{clusterchi}
Assumption~\ref{cluster} is closed with respect to smearing of fields $B$ with Schwartz functions; more precisely, it remains valid, with some other constants $c_i$,  under replacement $B_i\rightarrow B_i(\chi_i)$.
\end{pr}

Assumption~\ref{cluster} is fulfilled, in particular, for local and almost local fields, as shown in Proposition 12 in Appendix B of \cite{he14}.

\paragraph{} 
\label{hr2psi}
For sufficiently small $\eta$, there is
\begin{equation}\label{eq:two}
 \lim_{R\to\infty}B_{1,R,\eta}^{+}[f_1]^*B_{2,R,\eta}^{+}[f_2]\W=(2\pi)^2\big(\ti{f}_1(1,P)B_1\W,E_0 \ti{f}_2(1,P) B_2\W\big)\,\W\,.
\end{equation}
The projection of this equality onto $\W$ follows from Eq.\,\eqref{eq:lim_B_one_eta}. Since the energy transfer of the operator $B^+_{1,R}[f_1]^*$ is contained in $(-\infty,0\>$, it holds $B^+_{1,R}[f_1]^*\W\in E(\{0\})\Hc$, where $E(\cdot)$ is the spectral projection of the four-momentum operators. As $f_1=b_1^{(3)}$, we have
\begin{equation}
  B^+_{1,R}[f_1]^*\Omega=E(\{0\}) B^+_{1,R}[f_1]^*\Omega=-i E(\{0\}) [P^0, B^+_{1,R}[b_1^{(2)}]^*]\Omega=0
\end{equation}
and the operator $B^+_{1,R}[f_1]^*$ annihilates the vacuum. Therefore, the relation \eqref{eq:two} will be true, if
\begin{equation}\label{eq:two_ort}
 \lim_{R\rightarrow \infty}\big\|E_\W^\bot\big[B^+_{1,R,\eta}[f_1]^*,B^+_{2,R,\eta}[f_2]\big]\W\big\|=0\,.
\end{equation}
Note that $B^+_{j,R,\eta}[f] = B_j[g^\eta_R,f_+]$ where
\begin{equation}
 f_+(s,l)=\int \theta(\w) \ti{f}(\w,l) e^{-i\w s}\,d\w\,.
\end{equation}
In general, if $f =b^{(n)}$ for some $b\in \Sdeg{-2+n}_\ep$, $\ep>0$, then $f_+ \in \Sdeg{-2}_{n}$ (use Lemma~14 in Appendix A of \cite{he14'}). Thus, under our assumptions $f_+\in \Sdeg{-2}_{3}$.

Since  $\supp g\subseteq\<\tau_1,\tau_2\>\subset(0,\infty)$, we have
\begin{equation}
 \supp g^\eta_R\subseteq \<R_1,R_2\>=\<R+(\tau_1-1)w(R), R+(\tau_2-1)w(R)\>.
\end{equation}
Therefore, the identity \eqref{eq:two_ort} is the consequence of the following lemma which is proved in Appendix~\ref{app:proof_lemma_fock}.
\begin{lem}\label{lem:fock}
For sufficiently small $\eta$
\begin{equation}
 \lim_{\begin{smallmatrix}R\to\infty\\r_1,r_2\in\<R_1,R_2\>\end{smallmatrix}}\big\|E_\W^\bot\big[B_{1}[r_1,f_{1+}]^*,B_{2}[r_2,f_{2+}]\big]\W\big\|=0\,.
\end{equation}
\end{lem}

\paragraph{} 
The Fock structure of the scalar product of asymptotic states
\begin{equation}\label{eq:fock}
 \lim_{R\to\infty}(B^+_{1,R,\eta}[f_1]\ldots B^+_{k,R,\eta}[f_n]\W,B^+_{k+1,R,\eta}[f_{k+1}]\ldots B^+_{n,R,\eta}[f_n]\W)
\end{equation}
can be obtained by transferring the operators $B^+_{j,R,\eta}[f_j]$ from the left to the adjoints on the right, commuting them to the far right and using \eqref{eq:two_ort} (this technique is described thoroughly in \cite{dy05}). Note that to prove this we used neither the spectral assumption nor the locality of the fields. Thus the method might be applied to asymptotic states defined in both paragraphs~\ref{hrasst} and~\ref{hrasst_loc}.

\section{Conclusions and outlook}

In the setting of (in general) nonlocal fields satisfying some mild decay conditions we have established a link between their null asymptotic behavior on the one hand, and their energy-momentum spectral properties in a neighborhood of the lightcone, on the other. These properties include, in particular, the condition of infrared-regularity, i.e.\ appropriate vanishing in momentum space for $p=0$. If the standard asymptote exists as a limit, it defines a quantum field satisfying the wave equation -- again an infrared regular field. The IR-regularity of the problem is reflected in the unique decomposition of the limit fields into positive and negative energy-transfer parts, which therefore have physical interpretation of creators and annihilators of some particle-like, zero-mass excitations.

In the more specific context of vacuum representation the scheme was applied for the derivation of a nonlocal massless version of the Haag-Ruelle theory with the resulting Fock space of asymptotic states. Strictly local setting is a special case, which simplifies Buchholz's analysis.

We mention that the general scheme also works in a nonlocal algebraic model proposed earlier by one of us as a candidate for the description of long-range structure of quantum electrodynamics \cite{he98} (see also \cite{he11} for more information and references): IR-regular fields (in the sense defined here) present in this model may be reconstructed from their null asymptotes. However, the model contains more general fields satisfying wave equation, whose IR behavior, although remaining under control in this model, prevents the application of the methods described in the present paper. Thus physical interpretation of the model in terms of asymptotic particles is not complete. In our view this is also to be expected in prospective full quantum electrodynamics which would not arbitrarily cut infrared regime. Null asymptotic analysis of more IR-singular--nonlocal fields, both in the model, as on more general grounds, is thus an interesting problem for future investigations.

\section*{Acknowledgements}
Pawe{\l} Duch acknowledges the support of the Polish Ministry of Science and Higher Education, grant number 7150/E-338/M/2013.

\section*{Appendix}
\setcounter{section}{0}
\renewcommand{\thesection}{\Alph{section}}

\section{An estimate}\label{sec:convolution}

\begin{lem}\label{lem:bound_convolution}
Let $f_j\in C^\infty(\mR)$ be such that $|f_j(s)|\leq \con/(\la+|s|)^{1+\ep}$,  $j=1,2$. If $\ep>0$ then
\begin{equation}\label{eq:app_bound_ff}
 \Bigg|\int\limits_{|s_1-s_2|\geq S}\! f_1(s_1) f_2(s_2)\, ds_1ds_2\,\Bigg| \leq \frac{\con}{(\la+|S|)^{\ep}}\,.
\end{equation}
\end{lem}
\begin{proof}
In the region $|s_1-s_2|\geq S$ either $|s_1|\geq S/2$ or $|s_2|\geq S/2$, thus, the lhs of \eqref{eq:app_bound_ff} is bounded by
\begin{equation}
\int\limits_{|s_1|\geq S/2} \frac{\con\, ds_1ds_2}{(\la+|s_1|)^{1+\ep}(\la+|s_2|)^{1+\ep}} \leq \frac{\con}{(\la+|S|)^{\ep}}\,.
\end{equation}
\end{proof}

\section{On Fourier transform on a sphere}\label{regwave}

In this section the setting is the $3$-dimensional Euclidean space (positive metric). We denote $\vec{M}f(\vec{l})=\vl\times\vec{\p}f(\vl)$, where $\vec{\p}$ is the contravariant derivative vector with respect to $\vl$. Note that in positive oriented, orthonormal basis there is
$\dsp M^i=\tfrac{1}{2}\sum_{j,k=1}^3\vep^{ijk}L_{jk}$ and $\dsp\vec{M}\cdot\vec{M}=\tfrac{1}{2}\sum_{j,k=1}^3(L_{jk})^2=\Delta_S$ -- the Laplace operator on the unit sphere.
\begin{pr}
Let $f(\hn)$ be a smooth function on the sphere $S^2$. Then
\begin{gather}
 \int e^{-i\vq\cdot\hn}f(\hn)\,d\W(\hn)=\frac{2\pi i}{|\vq\,|}\Big(e^{-i|\vq|}f(\hq)-e^{i|\vq|}f(-\hq)\Big)+(\F_3R)(\vq)\,,\label{expS}\\[1ex]
 (\F_3R)(\vq)=i\int e^{-i\vq\cdot\hn}\frac{\vq\times\hn}{|\vq\times\hn|^2}\cdot\vec{M}f(\hn)d\W(\hn)\,,\label{Rq}\\[1ex]
 R(\vz)=\frac{-1}{|\vz|^2-1}\int\theta(\vz\cdot\hn-1)\Delta_Sf(\hn)d\W(\hn)\,,\label{Rz}
\end{gather}
It follows that
\begin{equation}\label{Rzest}
 |R(\vz)|\leq 2\pi\|\Delta_Sf\|_\infty\,\frac{\theta(|\vz|-1)}{|\vz|(|\vz|+1)}\,.
\end{equation}
\end{pr}
\begin{proof}
Let $(\vth,\vph)$ be the standard spherical angles for $\hn$ with respect to an orthonormal basis with the third vector along $\vq$.  In the integral on the lhs of Eq.\,\eqref{expS} we write $\sin\vth e^{-i\vq\cdot\hn}=(-i/|\vq|)\p_\vth e^{-i\vq\cdot\hn}$ and integrate by parts with respect to $\vth$. The boundary terms give the first two terms on the rhs of Eq.\,\eqref{expS} and the remaining integral has the form $(i/|\vq|)\int e^{-i\vq\cdot\hn}\p_\vth f(\hn)d\vth d\vph$. Now noting simple identity
$\p_\vth=\dfrac{\hq\times\hn}{|\hq\times\hn|}\cdot\vec{M}$ and using $|\hq\times\hn|=\sin\vth$ we arrive at \eqref{Rq}.

To calculate the inverse Fourier transform of this function we note that $\dfrac{\vq\times\hn}{|\vq\times\hn|^2}=\dfrac{\vq_\bot}{|\vq_\bot|^2}\times\hn$, where $\vq_\bot$ is the component of $\vq$ orthogonal to $\hn$. Using standard methods one finds that
\begin{equation}
 \frac{1}{(2\pi)^2}\int e^{i\vq\cdot\vz} e^{-i\vq\cdot\hn}\frac{\vq_\bot}{|\vq_\bot|^2}d^3q=i\delta(\vz\cdot\hn-1)\frac{\vz_\bot}{|\vz_\bot|^2}
 =i\delta(\vz\cdot\hn-1)\frac{\vz_\bot}{|\vz|^2-1}
\end{equation}
in the distributional sense. Therefore,
\begin{multline}
 R(\vz)=\frac{-1}{|\vz|^2-1}\int\delta(\vz\cdot\hn-1)(\vz\times\hn)\cdot\vec{M}f(\hn)d\W(\hn)\\
 =\frac{1}{|\vz|^2-1}\int\vec{M}\theta(\vz\cdot\hn-1)\cdot\vec{M}f(\hn)d\W(\hn)\,.
\end{multline}
Transferring $\vec{M}$ to $f$ we arrive at \eqref{Rz}.
\end{proof}

\section{Spaces $L^{p,1}$}\label{spacesp1}

For $p\geq1$, we define seminorms on the space of measurable functions on Minkowski space $M$
\begin{equation}\label{nnorm}
 \|\rho\|_{p,1}=\int\|\rho(x^0,.)\|_p\,dx^0\,,
\end{equation}
where the sign $\|.\|_p$ under the integral denotes the $L^p(\mR^3,d^3x)$-norm. If $h_3$ is a function on the $3$-space and $h_0$ a function of $x^0$, then the use of H\"older's inequality on $3$-space or in $x^0$ respectively shows that
\begin{equation}\label{30hold}
 \|h_3\rho\|_1\leq\|h_3\|_q\|\rho\|_{p,1}\,,\qquad \|h_0\rho\|_{p,1}\leq\|h_0\|_q\|\rho\|_p\,,
\end{equation}
where $q^{-1}+p^{-1}=1$. In particular, if $\rho=0$ almost everywhere, then by the second inequality above $\|h_0\rho\|_{p,1}=0$ for all characteristic functions of bounded sets $h_0$, which implies $\|\rho\|_{p,1}=0$. Conversely, if $\|\rho\|_{p,1}=0$, then by the first inequality $\|h_3\rho\|_1=0$ for all characteristic functions of bounded sets $h_3$, which implies $\rho=0$ almost everywhere. Thus classes of functions coinciding almost everywhere, with finite seminorms \eqref{nnorm}, form normed spaces, which we denote $L^{p,1}$.

Spaces $L^{p,1}$ are complete. Namely, let $\vph_n$ be a Cauchy sequence. By the first of inequalities \eqref{30hold} $h_3\vph_n$ is a Cauchy sequence in $L^1$ for any characteristic function of a bounded set $h_3$. Therefore, by completeness of $L^1$, there exists limit $\vph(x)=\lim_n\vph_n(x)$ almost everywhere. Choose $\ep>0$ and $N$ such that $\|\vph_k-\vph_n\|_{p,1}\leq\ep$ for all $k,n\geq N$. Then for $n\geq N$, by the use of Fatou's lemma, we have
\begin{multline}
 \|\vph-\vph_n\|_{p,1}=\int\Big(\int\lim_{k\to\infty}|(\vph_k-\vph_n)(x)|^pd^3x\Big)^{1/p}dx^0 \\
 \leq\int\Big(\varliminf_{k\to\infty}\int|(\vph_k-\vph_n)(x)|^pd^3x\Big)^{1/p}dx^0
 \leq\varliminf_{k\to\infty}\|\vph_k-\vph_n\|_{p,1}\leq\ep\,,
\end{multline}
which closes the proof.

For convolution of functions from spaces $L^{p,1}$ an analog of Young's inequality is true. Let $1+p^{-1}=q^{-1}+r^{-1}$. Then
\begin{equation}\label{young}
  \|\vph*\psi\|_{p,1}\leq\|\vph\|_{q,1}\|\psi\|_{r,1}\,.
\end{equation}
This is shown with the use of Young's inequality in $3$-space (${\stackrel{(3)}{*}}$ denotes the $3$-space convolution):
\begin{multline}
 \mathrm{lhs}=\int\Big\|\int\vph(x^0-y^0,.){\stackrel{(3)}{*}}\psi(y^0,.)dy^0\Big\|_pdx^0\\
 \leq\int\int\|\vph(x^0,.){\stackrel{(3)}{*}}\psi(y^0,.)\|_pdy^0dx^0
 \leq\int\|\vph(x^0,.)\|_qdx^0\int\|\psi(y^0,.)\|_rdy^0\\=\mathrm{rhs}\,.
\end{multline}
In particular, for $r=1$ we find
\be\label{young1}
 \|\vph*\psi\|_{p,1}\leq\|\vph\|_{p,1}\|\psi\|_1\,.
\ee
In similar way it is also easy to show that for the $3$-space convolution with a~function $h_3$ on the $3$-space there is
\be\label{young13}
\|\vph{\stackrel{(3)}{*}}h_3\|_{p,1}\leq\|\vph\|_{p,1}\|h_3\|_1\,.
\ee

\section{Proof of the estimate \eqref{chiest} of $\|\chi_R^\pm\|_{q,1}$}\label{nest}

We represent $\wh{\chi_R^\pm}(p)=\wh{\vph^\pm_R}(p)\wh{\rho^\pm_R}(p)$, with
\begin{align}
 \wh{\vph^\pm_R}(p)&=\theta(\pm p^0)|p^0|^\delta|\vp|^{-1} \ti{g}(Rp^\pm)\big(1+R^2|\vp|^2\big)^2\,,\\
 \wh{\rho^\pm_R}(p)&=\ti{f}(p^0,\hp^\mp)\big(1+R^2|\vp|^2\big)^{-2}\,.
\end{align}
For $\chi^\pm_R=(2\pi)^{-2}\vph^\pm_R*\rho^\pm_R$ one obtains with the use of inequality \eqref{young}:
\begin{equation}\label{chfr}
 \|\chi^\pm_R\|_{q,1}\leq(2\pi)^{-2}\|\vph^\pm_R\|_{u,1}\|\rho^\pm_R\|_{v,1}\,,
\end{equation}
where $q=6/(6-\kappa')$, $u=15/(12-\kappa')$ and $v=10/(12-\kappa')$; as $\kappa'>2$, we have $u>3/2$ and $v>1$. The inverse transforms of $\wh{\vph^\pm_R}$ and $\wh{\rho^\pm_R}$ may be written as
\be
 \vph^\pm_R(x)=R^{-3-\delta}\frac{1}{(2\pi)^2}\int\theta(\pm q^0)|q^0|^\delta|\vec{q}|^{-1}\ti{g}(\pm(|q^0|+|\vec{q}|))(1+|\vec{q}|^2)^2e^{-iq\cdot(x/R)}dq\,,
\ee
\be
 \rho^\pm_R(x)=R^{-3}\frac{1}{(2\pi)^2}\int f(x^0,t\mp\hat{q})(1+|\vec{q}|^2)^{-2}e^{i\vec{q}\cdot(\vx/R)}d^3q\,.
\ee
We now apply to these transforms Lemma~14 in Appendix A in \cite{he14'}: for $\vph^\pm_R(x)$ separately\footnote{To show that this is possible one needs some work. The crucial property which enables such use of the lemma, is that singularities of derivatives of $\vph_R(q)$ factorize in $q^0$ and $|\vq|$ and independently satisfy the assumptions of this lemma.} in $q^0$ and $\vec{q}$, and for $\rho^\pm_R(x)$ in $\vec{q}$. Moreover, we use the decay rate of $f(s,l)$ (and its intrinsic derivatives on the cone) in $s$. In this way we obtain the estimates
\begin{align}
 |\vph^\pm_R(x)|&\leq \frac{\con\, R^{-3-\delta}}{(1+(|x^0|/R))^{1+\delta}(1+(|\vx|/R))^2}\,,\\
 |\rho^\pm_R(x)|&\leq\frac{\con\,R^{-3}}{(1+|x^0|)^{1+\ep}(1+(|\vx|/R))^3}\,.
\end{align}
Using this in \eqref{chfr}, one obtains $\|\chi^\pm_R\|_{q,1}\leq \con\,R^{-(\kappa'-2)/2-\delta}$.

\section{Spectral condition in free theory}
\label{app:spectral_condition_weyl}

\begin{pr}
Let $\Hc$ be the bosonic Fock space of positive energy vacuum representation of free massless field, with the single-particle Hilbert space $\Hc_1$ equipped with the standard scalar product
\begin{equation}
 (J_1,J_2) = \frac{1}{(2\pi)^3}\int \theta(p^0) \delta(p^2) \overline{\wh{J}_1(p)} \wh{J}_2(p)\, d^4 p\,.
\end{equation}
and creation/anihilation operators denoted by $a^*(J)$ and $a(J)$.
Then for any Schwartz function $J$
\begin{equation}
 \|(E_\mu-E_0)W(J) \W\| = O(\mu)\,,
\end{equation}
where $\W$ is the vacuum state, $W(J)=e^{-i (a(J) +a^*(J))}$ is the Weyl operator and projections $E_\mu$ were defined in paragraph~\ref{hrspec} of Section~\ref{hrtheory}.
\end{pr}
\begin{proof}
Let us note that
\begin{equation}
 W(J) \W =  \exp(-\tfrac{1}{2}(J,J)) \exp(-i a^*(J)) \W\,.
\end{equation}
It holds
\begin{multline}
 \|(E_\mu-E_0)\exp(-i a^*(J)) \W\|
 \\
 =\|(E_\mu-E_0)(\exp(-i a^*(J)) + i a^*(J) - 1)\W\|
 \\
 \leq\| E_\mu  \phi(a^*(J)) a^*(J)^2\W\| \leq \| \phi(a^*(J)) P_{\Hc_{2}}\| \| E_\mu   a^*(J)^2\W\|
\end{multline}
where $\phi(x)=(\exp(-ix)-1+ix)/x^2$. The operator $\phi(a^*(J))$ is defined in terms of power series expansion of $\phi$ on the two-particle subspace $\Hc_{2}$ of the Fock space $\Hc$ and $P_{\Hc_{2}}$ is projection onto $\Hc_{2}$. The operator $\phi(a^*(J)) P_{\Hc_{2}}$ is bounded because $\|a^*(J)^n P_{\Hc_{2}}\|^2=\frac{(n+2)!}{2}(J,J)^n$. Let us also observe that $E_\mu \phi(a^*(J)) P_{\Hc_{2}}(\id- E_\mu)=0$ since $\phi(a^*(J))\Hc_{2}$ has energy-momentum transfer in $\clc$.

The thesis now follows from the estimate
\begin{multline}
\| E_\mu A(J)^2\W\|^2 = \frac{1}{(2\pi)^6}\int\limits_{2p_1\cdot p_2 \leq\mu^2} |\wh{J}(|\vec{p}_1|,\vec{p}_1)|^2 |\wh{J}(|\vec{p}_2|,\vec{p}_2)|^2\,\frac{d^3 p_1}{2|\vec{p}_1|} \frac{d^3 p_2}{2|\vec{p}_2|}
\\
\leq \con\, \int\limits_{0\leq  \eta \leq\mu^2/(4\w_1\w_2)} \frac{\w_1d\w_1\,\w_2d\w_2\,d\eta}{(1+\w_1)^2(1+\w_2)^2}   \leq \con\,\mu^2\,,
\end{multline}
where $p_j=(|\vec{p}_j|,\vec{p}_j)$, $j=1,2$ and in the second line we have introduced new integration variables $\omega_j=|\vec{p}_j|$,
 $\eta=p_1\cdot p_2/(2\omega_1\omega_2)$ and used the estimate  $|\wh{J}(p)|\leq \con\,(\la+|\vec{p}|)^{-2}$.
\end{proof}

Since Weyl operators are not smooth (in fact they are not even differentiable), we could not use them in the construction in Section~\ref{hrtheory}. However, the above result is also true for the smooth operator $B=W(J)(\chi)$ obtained by smearing Weyl operator with arbitrary Schwartz function $\chi$. It immediately follows that operator $B$ fulfils the spectral condition formulated in paragraph~\ref{hrspec} of Section~\ref{hrtheory}.

\section{Proof of Lemma \ref{lem:fock}}
\label{app:proof_lemma_fock}

We denote $r^2=r_1r_2$, $\Delta r=r_2-r_1$, $\Delta s_{jk}=s_j-s_k$, $\xi^2_{jk}=l_j\cdot l_k/2$ and observe that $|\Delta r|\leq\tau w(R)$, where $\tau=\tau_2-\tau_1$. It is easy to see that for sufficiently large~$R$ there is $w(R)\leq 2w(r)\equiv 2w$, and then $|\Delta r|\leq 2 \tau w$, what we assume from now on.  We have
\begin{multline}
 I = \| E_\Omega^\perp [B_{1}[r_1,f_{1+}]^*,B_{2}[r_2,f_{2+}]] \Omega \| \\[1ex]
 = \frac{r^2}{4\pi}\Big\| \int E_\Omega^\perp [B^*_1(s_1 t + r_1 l_1),B_2(s_2 t + r_2 l_2)] \Omega\, \prod_{i=1}^2f_{i+}(s_i,l_i)d s_i d\Omega_t(l_i) \Big\|\,.
\end{multline}
As $f_{j+}\in \Sdeg{-2}_{3}$, there is $|f_{j+}(s,l)|\leq \con\,(\la+|s|)^{-4}$ for $t\cdot l=1$, $j=1,2$. To prove the thesis of Lemma, we split the above integral into the integrals over the following regions:\\
(i) $r\xi_{12} \geq \tau w$.
In this region $I_{(i)}\leq  \con\, w^{-(\kappa-2)}$ -- this is easily shown by adapting the proof of case (iii) of Theorem~\ref{limit}, letting $\xi_0=\tau w/r$, $d=2\tau w/r$.\\
(ii) $r\xi_{12} < \tau w$,  $|\Delta s_{12}| \geq \tau  w$.
 This contribution is bounded by
 \begin{equation}
  I_{(ii)}\leq \con\, r^2 \int\limits_0^{(\tau w/r)^2}\! d \xi^2_{12}  \int\limits_{|\Delta s_{12}|\geq \tau w} \frac{d s_1 d s_2}{(\la+|s_1|)^{4}(\la+|s_2|)^{4}} \leq \con\, w^{-1},
 \end{equation}
 where we used Lemma~\ref{lem:bound_convolution} to estimate the $d s_1 d s_2$ integral.\\
(iii) $r\xi_{12} < \tau w$, $|\Delta s_{12}| < \tau w$. Here we consider the squared contribution $I^2_{(iii)}$ and note that
\[
 I^2_{(iii)}\leq
 \con\, r^4\int K(y_1,y_2,y)\prod_{i=1}^4\frac{ds_i\, d\Omega_t(l_i)}{(\la+|s_i|)^{4}},
\]
where
\begin{gather*}
 y_1 = \Delta s_{12} t + r_1 l_1 - r_2 l_2\,,\quad y_2 = \Delta s_{34} t + r_1 l_3 - r_2 l_4\,,\\
 y=\frac{1}{2}((\Delta s_{13} + \Delta s_{24})t + r_1(l_1-l_3) + r_2(l_2-l_4))
\end{gather*}
and the above integral is over the region
\begin{equation}\label{setiii}
 r\xi_{12} < \tau w\,,\ r\xi_{34} < \tau w\,,\ |\Delta s_{12}|<\tau w\,,\ |\Delta s_{34}|<\tau w\,.
\end{equation}
Therefore, we can estimate further
\begin{equation}
 I^2_{(iii)}\leq
 \con\,w^4 \int \sup_{l_2,l_3}|K(y_1,y_2,y)|\prod_{i=1}^4\frac{ds_i}{(\la+|s_i|)^{4}}\,d\Omega_t(l_1)d\Omega_t(l_4)\,,
\end{equation}
where supremum is over the set restricted by the first two relations in \eqref{setiii}.
We split the set determined by \eqref{setiii} further into regions:\\
(a) $|y^0|\geq 2\tau w$. We have $|\Delta s_{13} + \Delta s_{24}| > 4 \tau w$ and $|\Delta s_{13} - \Delta s_{24}| < 2 \tau w$. Thus, $|\Delta s_{13}|,|\Delta s_{24}|\geq \tau w$, and estimating $K(y_1,y_2,y)$ by a constant we find
\begin{equation*}
 I^2_{(a)}\leq \con\, w^4 \int\limits_{\begin{smallmatrix}|\Delta s_{13}|\geq\tau w\\|\Delta s_{24}|\geq \tau w\end{smallmatrix}} \prod_{i=1}^4\frac{ds_i}{(\la+|s_i|)^{4}}
 \leq \con\, w^{-2}.
\end{equation*}
The last inequality follows from Lemma~\ref{lem:bound_convolution} applied to the $d s_1 d s_3$ and $d s_2 d s_4$ integrals.\\
(b) $|y^0|<2\tau w$, $|\vec{y}|\leq (C+2) \tau w$, where $C$ is some constant to be fixed later. Since
\begin{equation}\label{eq:fock_vec_y^2}
 |\vec{y}|^2 = r_1^2 \xi_{13}^2 + r_2^2 \xi_{24}^2+r_1 r_2 (\xi_{14}^2 + \xi_{32}^2 -\xi_{12}^2 -\xi_{34}^2)\,,
\end{equation}
in this region $r^2\xi_{14}^2 \leq |\vec{y}|^2  +r_1 r_2 (\xi_{12}^2 +\xi_{34}^2) \leq (C+3)^2\tau^2 w^2\equiv \al^2w^2$ and
 \begin{equation*}
 I^2_{(b)}\leq \con \,w^4\int\limits_0^{(\al w/r)^2}\!d\xi^2_{14}
 \leq \con\, \frac{w^6}{r^2}\,.
\end{equation*}
This vanishes in the limit $r\rightarrow\infty$ if $\eta<1/3$.\\
(c) $|y^0| < 2 \tau w$, $|\vec{y}|>(C+2) \tau w$.
In this region $|\vec{y}| - |y^0| \geq C\tau w$.
Using the identities $|\vec{y}_1|^2 = (\Delta r)^2 + 4 r^2 \xi_{12}^2$,
$|\vec{y}_2|^2 = (\Delta r)^2 + 4 r^2 \xi_{34}^2$ we obtain estimates
$|\vec{y}_j|\leq 3 \tau w$ and $|y_j| \leq |y_j^0| + |\vec{y}_j| \leq 4\tau w$, $j=1,2$.
Therefore, for $C\geq4c_1$ the terms of Assumption~\ref{cluster} are satisfied with $d\equiv d(r)=4\tau w(r)$ and we can use the estimate~\eqref{clusterest}. Moreover, using \eqref{eq:fock_vec_y^2} we find for $C\geq4$:
\[
 4|\vec{y}|^2\geq 2\left(|\vec{y}|^2 - r^2 (\xi_{12}^2 +\xi_{34}^2) + r^2 \xi_{14}^2 \right) > 2(5^2 \tau^2w^2 + r^2 \xi_{14}^2)
  \geq(5\tau w + r \xi_{14})^2\,,
\]
which implies $|\vec{y}| - |y^0|\geq \frac{1}{2}(\tau w + r \xi_{14})$.
Thus, for $C\geq\max\{4,4c_1\}$ we obtain the estimate
\begin{multline}
 I^2_{(c)}\leq \con\, w^4 \int_{0}^1 d \xi^2_{14}
 \left( c_2\, \frac{d^M}{\left( \frac{1}{2} (\tau w + r \xi_{14})\right)^{\nu}} + c_3 d^{-N} \right)
 \\
 \leq w^4 \left(\con\, w^M r^{-\nu'}+ \con\, w^{-N} \right)\,,
\end{multline}
where $\nu'=\min\{\nu,2\}$. $I^2_{(c)}$ vanishes in the limit if $\eta<\frac{\nu'}{M+4}$ and $N>4$.

\end{document}